\documentclass[journal]{IEEEtran}

\usepackage{amsmath,amssymb,amsfonts}
\usepackage{graphicx}
\usepackage{textcomp}
\usepackage{xcolor}
\usepackage{cite}
\usepackage{bm}
\usepackage{booktabs}
\usepackage{multirow}
\usepackage{array}
\usepackage{balance}
\usepackage{url}
\usepackage{algorithm}
\usepackage{algorithmic}
\newcommand{\E}{\mathbb{E}}
\newcommand{\R}{\mathbb{R}}
\newcommand{\C}{\mathbb{C}}
\newcommand{\CN}{\mathcal{CN}}
\newcommand{\calA}{\mathcal{A}}
\newcommand{\calF}{\mathcal{F}}
\newcommand{\btheta}{\boldsymbol{\theta}}
\newcommand{\bJ}{\mathbf{J}}
\newcommand{\bI}{\mathbf{I}}
\newcommand{\xiq}{x_{\mathrm{IQ}}}
\newcommand{\CRB}{\mathrm{CRB}}
\newcommand{\DR}{\mathrm{DR}}

\newcommand{\diag}{\mathrm{diag}}
\renewcommand{\Re}{\mathrm{Re}}
\renewcommand{\Im}{\mathrm{Im}}

\newtheorem{theorem}{Theorem}
\newtheorem{proposition}{Proposition}
\newtheorem{corollary}{Corollary}
\newtheorem{lemma}{Lemma}
\newtheorem{remark}{Remark}

\begin{document}

\title{Fisher Information Limits of Satellite RF Fingerprint Identifiability for Authentication}

\author{Haofan~Dong,~\IEEEmembership{Student~Member,~IEEE,}
        and~Ozgur~B.~Akan,~\IEEEmembership{Fellow,~IEEE}
\thanks{H.\ Dong and O.\ B.\ Akan are with the Internet of Everything Group, 
Department of Engineering, University of Cambridge, CB3 0FA Cambridge, U.K.(e-mail: hd489@cam.ac.uk, oba21@cam.ac.uk)}
\thanks{O.\ B.\ Akan is also with the Center for neXt-generation 
Communications (CXC), Department of Electrical and Electronics Engineering, 
Ko\c{c} University, 34450 Istanbul, Turkey.}}

\markboth{IEEE Transactions on Information Forensics and Security}%
{Dong and Akan: Fundamental Limits of Satellite RF Fingerprinting}

\maketitle

\begin{abstract}
RF fingerprinting authenticates satellite transmitters by exploiting hardware-specific signal impairments, yet existing methods operate without theoretical performance guarantees. We derive the Fisher information matrix (FIM) for joint estimation of IQ imbalance and power amplifier (PA) nonlinearity parameters, establishing Cram\'{e}r-Rao bounds (CRBs) whose structure depends on constellation moments. A necessary condition for full IQ identifiability is $\beta = 1 - |\E[x^2]|^2 > 0$; setting $\beta = 0$ for BPSK yields a rank-deficient FIM, showing that IQ imbalance parameters are unidentifiable from binary-modulated signals. This provides a plausible theoretical explanation for OrbID's AUC~$= 0.53$ on Orbcomm. From the FIM, we define an identifiability-induced discrimination metric that predicts which hardware parameters dominate authentication for a given modulation format. For constant-modulus PSK signals, the analysis predicts that PA nonlinearity features dominate while IQ imbalance features are ineffective. We validate the framework on 24~Iridium satellites using two recording campaigns, achieving cross-file fingerprint correlation $r = 0.999$ for PA features and confirming all four CRB predictions against empirical discrimination ratios. A DR-weighted authentication test achieves AUC~$= 0.934$ from six features versus $0.807$ with equal weighting, outperforming ML classifiers (AUC~$\leq 0.69$) on the same data.

\end{abstract}

\begin{IEEEkeywords}
Cram\'{e}r-Rao bound, Fisher information, RF fingerprinting, satellite authentication, hardware impairments, IQ imbalance, power amplifier nonlinearity, physical-layer security.
\end{IEEEkeywords}

\section{Introduction}
\label{sec:intro}

Low Earth orbit (LEO) mega-constellations form critical infrastructure for global communications, navigation, and remote sensing. An adversary who impersonates a legitimate
satellite can inject false data, intercept traffic, or disrupt
service~\cite{Yue2023LEOSecurity,Lohan2021GNSSSpoofing}. Cryptographic authentication mitigates
these threats in principle, but key management across hundreds of orbiting
nodes remains an open challenge~\cite{danev2010attacks}. RF fingerprinting offers a complementary defense by exploiting the unique hardware signature that each satellite's RF chain imprints onto the transmitted signal~\cite{polak2011identifying}.

Recent deep-learning methods have demonstrated satellite RF 
fingerprinting on Iridium~\cite{smailes2023watchthisspace,smailes2025satiq,oligeri2023pastai}, 
yet OrbID~\cite{solenthaler2025orbid} obtains near-random performance 
(AUC~$=0.53$) on Orbcomm with no explanation for the failure. Recent TIFS work has advanced RFFI modeling~\cite{Zhang2021RFFI,Rajendran2022RFImpairment} and robustness~\cite{Shen2022LoRa,Shen2023LengthVersatile,Luo2025ChannelRobust5G}, but these methods share a common limitation: no estimation-theoretic framework connects signal parameters to authentication performance.
Zhang et al.~\cite{zhang2025fingerprinting_survey} identify the absence of verifiable
accuracy guarantees as a persistent gap, echoing a call by
Polak et al.~\cite{polak2011identifying} over a decade ago.
In related domains, CRB-based analysis has been used to link
parameter estimation accuracy to authentication performance:
Zhang et al.~\cite{Zhang2021PhaseNoisePLA} derived closed-form detection
probabilities from channel/phase-noise estimation covariance in massive MIMO,
and Liu et al.~\cite{Liu2023mmWavePLA} extended this approach to mmWave
using the posterior Bayesian CRB. However, no prior work derives the CRB
for hardware impairment parameters in the context of RF fingerprinting,
nor establishes modulation-dependent conditions determining whether
fingerprinting is feasible or impossible.

This paper makes four contributions:
\begin{enumerate}
\item A closed-form FIM for joint IQ imbalance and PA nonlinearity estimation (Theorem~\ref{thm:fim}), revealing the identifiability factor $\beta = 1 - |\E[x^2]|^2$ as a necessary condition for IQ parameter recovery.

\item A rank-deficiency result for binary modulations (Proposition~\ref{prop:bpsk}): the FIM has rank at most~2, rendering IQ parameters unidentifiable and providing a plausible explanation for OrbID's AUC~$= 0.53$ on Orbcomm.

\item A CRB-induced discrimination metric that predicts which hardware parameters dominate authentication, validated on 24~Iridium satellites (cross-file PA correlation $r = 0.999$, all four predictions confirmed).

\item An identifiability-weighted authentication test (Algorithm~\ref{alg:iwat}) achieving AUC~$= 0.934$ via DR$^2$-weighted scoring, outperforming equal weighting ($0.807$) and ML baselines ($\leq 0.69$).
\end{enumerate}

Section~\ref{sec:related} reviews related work. Sections~\ref{sec:model}--\ref{sec:crb} present the system model and CRB derivation. Section~\ref{sec:ident} establishes identifiability conditions and the discrimination bound. Sections~\ref{sec:estimation}--\ref{sec:experiments} describe estimation, experimental validation, and the proposed authentication algorithm. Section~\ref{sec:discussion} discusses implications, and Section~\ref{sec:conclusion} concludes.

\section{Related Work}
\label{sec:related}

\begin{figure*}[t]
\centering
\includegraphics[width=0.85\textwidth]{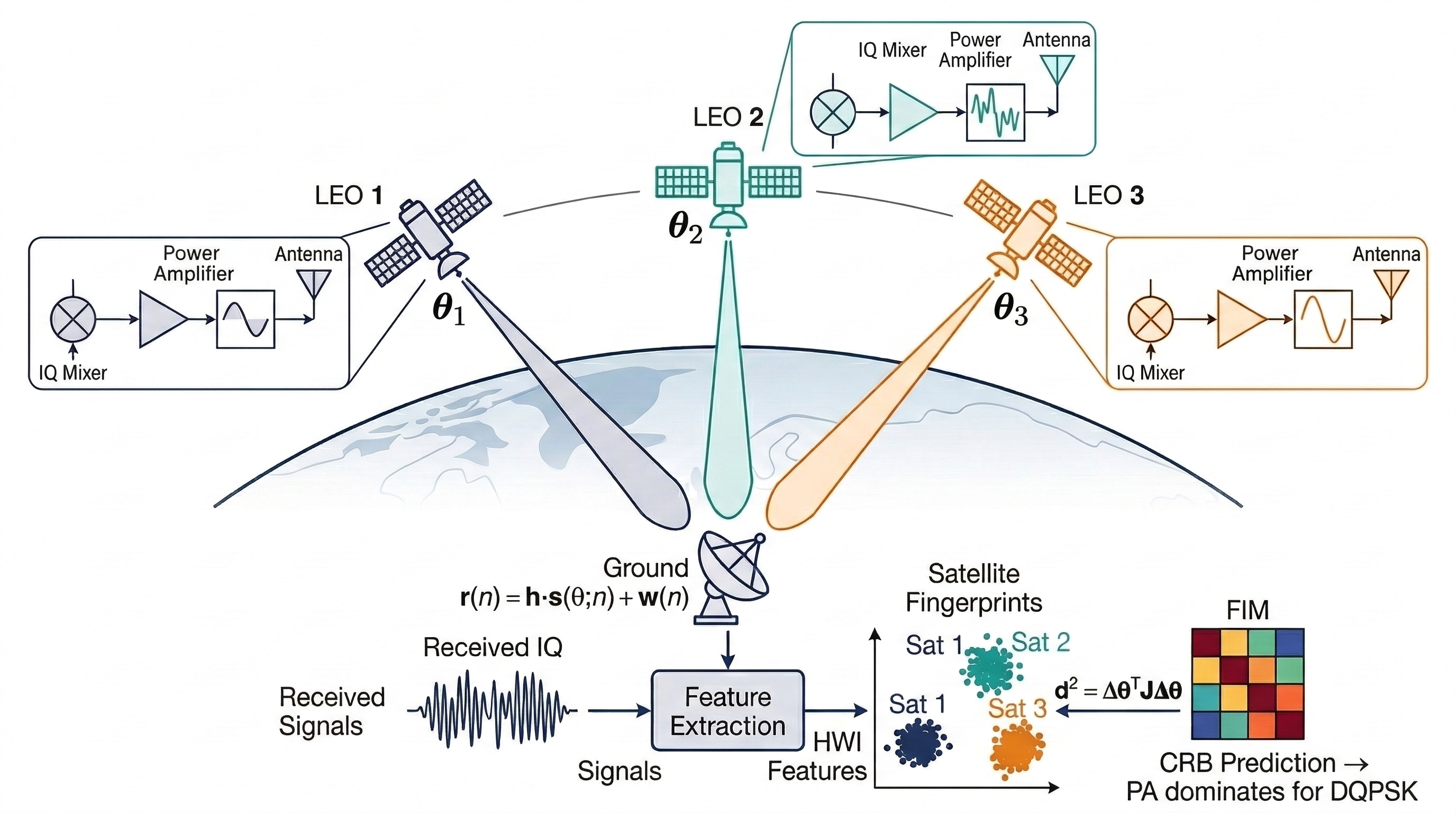}
\caption{System overview of CRB-based satellite RF fingerprinting. Each satellite's RF chain (IQ mixer $\to$ PA $\to$ antenna) introduces hardware-specific distortions parameterized by $\btheta_k = [\varepsilon_k, \varphi_k, \Re(\alpha_{3,k}), \Im(\alpha_{3,k})]^T$. The ground receiver extracts HWI features from the received IQ signal and uses the FIM-induced discrimination metric to authenticate transmitters.}
\label{fig:system}
\end{figure*}

\subsection{RF Fingerprinting}

RF fingerprinting exploits manufacturing-induced hardware imperfections to identify transmitters without cryptographic credentials~\cite{polak2011identifying}. Terrestrial systems have been demonstrated for WiFi~\cite{sankhe2019oracle}, LoRa~\cite{robyns2017physical,Shen2022LoRa,Shen2023LengthVersatile}, ADS-B~\cite{jian2020deep}, and IoT devices~\cite{jagannath2022comprehensive,xu2016device}, with recent TIFS work advancing analytical HWI models~\cite{Zhang2021RFFI,Rajendran2022RFImpairment} and channel-robust DL-based fingerprinting~\cite{Luo2025ChannelRobust5G}. A common limitation is the absence of theoretical performance bounds: no framework predicts whether a given hardware parameter is identifiable or which parameters dominate discrimination. Senigagliesi et al.~\cite{senigagliesi2021comparison} showed that statistical hypothesis testing can match or exceed ML for channel-based authentication, providing precedent for parametric approaches.

Satellite RF fingerprinting poses distinct challenges: Rician fading (K-factor 10--20~dB), Doppler shifts up to $\pm 36$~kHz, and tighter manufacturing tolerances~\cite{3gpp_tr38811,loo1985statistical}. SatIQ~\cite{smailes2023watchthisspace,smailes2025satiq} achieves AUC~$= 0.960$ on 66~Iridium satellites using oversampled IQ~\cite{smailes2024stickyfingers}. PAST-AI~\cite{oligeri2023pastai} trains a CNN on symbol-rate IQ from gr-iridium~\cite{gr-iridium}. OrbID~\cite{solenthaler2025orbid} extends SatIQ to Orbcomm (SDPSK) but achieves only AUC~$= 0.53$, a failure unexplained by the authors. Topal and Karabulut Kurt~\cite{Topal2022LEOPLA} proposed Doppler-based PLA, but Doppler features are extrinsic (channel-dependent) rather than intrinsic (hardware-dependent).

\subsection{CRB for Hardware Impairment Estimation}

The Cram\'{e}r-Rao bound provides a fundamental lower limit on the variance of any unbiased estimator~\cite{kay1993fundamentals,vantrees1968detection,cramer1946mathematical}. CRBs have been derived for channel estimation under hardware impairments~\cite{schenk2008rf,boulogeorgos2016effects}, IQ imbalance compensation~\cite{anttila2008circularity}, and PA model identification~\cite{saleh1981frequency,rapp1991effects}. However, these works treat
hardware impairments as nuisance parameters to be compensated, not as
fingerprint features to be discriminated.

In the authentication context, hypothesis-testing-based PLA
has been studied for channel features:
Xie et al.~\cite{Xie2021MultiFeaturePLA} combined multiple channel-based
features via composite hypothesis testing;
Zhang et al.~\cite{Zhang2021PhaseNoisePLA} derived closed-form detection
probabilities from channel gain and phase noise estimation covariance
in massive MIMO; and Liu et al.~\cite{Liu2023mmWavePLA}
linked the posterior Bayesian CRB to mmWave PLA performance.
Hanna and Cabric~\cite{Hanna2019PAFingerprint} demonstrated that PA
nonlinearity coefficients vary sufficiently across devices for
identification, but did not derive estimation bounds.
These works establish the paradigm of connecting
parameter estimation precision to authentication performance~\cite{Kay1998Detection},
following the classical identifiability framework~\cite{Lehmann1998PointEstimation,Stoica2005Spectral}.

The recent TIFS survey by Zhang et al.~\cite{zhang2025fingerprinting_survey} catalogues signal models and DL architectures but identifies no estimation-theoretic performance bounds for RF fingerprinting---the gap addressed in this paper.

\subsection{Threat Model and Adversary Capabilities}
\label{sec:threat}

LEO mega-constellations face an expanding threat surface
including eavesdropping, jamming, and identity
spoofing~\cite{Yue2023LEOSecurity}.
We consider a physical-layer authentication scenario in
which a ground verifier must determine whether a received
satellite signal originates from a legitimate transmitter or an
adversary.

\emph{Adversary model.} The adversary controls a transmitter capable of generating arbitrary baseband signals at the satellite frequency band. The adversary may (i)~eavesdrop on legitimate satellite transmissions, (ii)~obtain commercially available RF components of the same type, and (iii)~attempt to replicate a specific satellite's fingerprint~$\btheta_k$. However, the adversary \emph{cannot} physically access the space-qualified RF chain of a specific satellite to measure its exact hardware parameters. This reflects the practical reality that orbital assets are physically inaccessible.

\emph{Attack taxonomy.}
We focus on impersonation attacks ($\btheta_{\mathrm{adv}} \neq \btheta_k$),
where detection difficulty scales with~$d^2$ (Theorem~\ref{thm:disc}).
Replay and hardware-cloning attacks are outside the present scope.
The CRB predicts the minimum detectable difference:
$|\Delta\theta_i| > \sqrt{\CRB(\theta_i)}$ suffices for discrimination.

\section{System Model}
\label{sec:model}

\subsection{Transmitted Signal}

We consider a satellite downlink transmission where the baseband signal at the output of an ideal transmitter is
\begin{equation}\label{eq:tx_signal}
s(n) = \sum_{k} a_k \, g(n - kT_s), \quad n = 0, 1, \ldots, N-1,
\end{equation}
where $a_k \in \calA$ denotes the $k$-th transmitted symbol drawn from constellation~$\calA$, $g(\cdot)$ is the pulse-shaping filter, and $T_s$ is the symbol period. For Iridium, the modulation is DQPSK with $\calA = \{(\pm 1 \pm j)/\sqrt{2}\}$, the pulse shape is a square-root raised cosine (SRRC) with roll-off factor $\alpha = 0.4$, and the symbol rate is $R_s = 25$~ksps~\cite{maine1995iridium,luo2023ambiguity}.

Each Iridium Ring Alert (IRA) burst contains $N_{\mathrm{pre}} = 64$ unmodulated preamble symbols followed by $N_{\mathrm{uw}} = 12$ BPSK unique-word symbols with the fixed pattern \texttt{0x789}, giving $N_{\mathrm{known}} = 76$ known symbols per burst.

\subsection{Hardware Impairment Model}

The transmitted signal passes through the satellite's RF chain, which introduces hardware-specific impairments. We model two dominant sources: IQ mixer imbalance and power amplifier (PA) nonlinearity.

\emph{IQ imbalance.} An imperfect IQ mixer produces the distorted signal
\begin{equation}\label{eq:iq_model}
\xiq(n) = K_1 \, x(n) + K_2 \, x^{*}(n),
\end{equation}
where $x(n)$ denotes the ideal baseband symbol, and the IQ imbalance coefficients are~\cite{schenk2008rf,anttila2008circularity}
\begin{equation}\label{eq:K1K2}
K_1 = \frac{1 + (1+\varepsilon) e^{j\varphi}}{2}, \quad K_2 = \frac{1 - (1+\varepsilon) e^{-j\varphi}}{2},
\end{equation}
with $\varepsilon$ denoting the gain imbalance and $\varphi$ denoting the phase imbalance. For an ideal mixer, $\varepsilon = 0$ and $\varphi = 0$, yielding $K_1 = 1$ and $K_2 = 0$.

\emph{PA nonlinearity.} The PA output is modeled by a memoryless polynomial~\cite{rapp1991effects,saleh1981frequency}
\begin{equation}\label{eq:pa_model}
y(n) = \alpha_1 \, \xiq(n) + \alpha_3 \, |\xiq(n)|^2 \, \xiq(n),
\end{equation}
where $\alpha_1$ is the linear gain (absorbed into the channel coefficient below) and $\alpha_3$ is the third-order nonlinearity coefficient. The fifth-order term $\alpha_5 |\xiq|^4 \xiq$ is omitted because $|\xiq(n)|^2$ is approximately constant for PSK, making $\alpha_5$ and $\alpha_3$ confounded (Jacobian column correlation exceeding 0.999). For QAM, the confounding is reduced; the QAM predictions assume the third-order model.

\subsection{Received Signal Model}

The received signal after propagation through a flat-fading satellite channel is
\begin{equation}\label{eq:rx_signal}
r(n) = h \left[ K_1 x(n) + K_2 x^{*}(n) + \alpha_3 |\xiq(n)|^2 \xiq(n) \right] + w(n),
\end{equation}
where $h \in \C$ is the channel coefficient (incorporating path loss, Rician fading, and linear gain $\alpha_1$) and $w(n) \sim \CN(0, \sigma^2)$ is additive white Gaussian noise. The parameter vector to be estimated is
\begin{equation}\label{eq:theta}
\btheta = [\varepsilon, \, \varphi, \, \Re(\alpha_3), \, \Im(\alpha_3)]^T \in \R^4.
\end{equation}
These four parameters constitute the hardware impairment \emph{fingerprint} of a satellite transmitter. Since each satellite has a unique RF chain manufactured with distinct component tolerances, the vector~$\btheta$ differs across satellites.

\subsection{Satellite Channel Characteristics}

The Iridium constellation operates at L-band (1616--1626.5~MHz) from LEO at approximately 780~km altitude~\cite{maine1995iridium}. The downlink channel exhibits several characteristics relevant to fingerprint estimation.

The line-of-sight component dominates, with Rician K-factors typically 10--20~dB for elevation angles above $20^{\circ}$~\cite{loo1985statistical,3gpp_tr38811}. LEO orbital velocity produces Doppler shifts up to $\pm 36$~kHz; for estimation purposes, the carrier frequency offset is treated as a nuisance parameter removed prior to fingerprint extraction.

The Iridium signal bandwidth (35~kHz at $R_s = 25$~ksps with $\alpha = 0.4$) is far below the coherence bandwidth of the satellite channel, so frequency-selective fading is negligible. The channel reduces to a single complex coefficient~$h$ per burst.

\section{Cram\'{e}r-Rao Bounds for HWI Parameter Estimation}
\label{sec:crb}

\subsection{Fisher Information Matrix}

For $N$ known transmitted symbols $\{x(n)\}_{n=0}^{N-1}$ observed at signal-to-noise ratio $\gamma = |h|^2/\sigma^2$, the $(i,j)$-th entry of the $4 \times 4$ Fisher information matrix is~\cite{kay1993fundamentals}
\begin{equation}\label{eq:fim_entry}
[\bJ(\btheta)]_{ij} = \frac{2}{\sigma^2} \sum_{n=0}^{N-1} \Re\left\{ \frac{\partial f^{*}(\btheta; n)}{\partial \theta_i} \cdot \frac{\partial f(\btheta; n)}{\partial \theta_j} \right\},
\end{equation}
where $f(\btheta; n) = h[K_1 x(n) + K_2 x^{*}(n) + \alpha_3 |\xiq(n)|^2 \xiq(n)]$ is the noise-free received signal.

For i.i.d.\ uniform symbols, the summation converges to constellation expectations. Define the \emph{complementary second moment} $\mu_{20} = \E[x^2]$, the \emph{fourth moment} $\mu_4 = \E[|x|^4]$, and the \emph{sixth moment} $\mu_6 = \E[|x|^6]$, where expectations are uniform over~$\calA$. For constant-modulus constellations (PSK), $\mu_6 = \mu_4 = 1$.

\begin{theorem}[Closed-Form FIM]\label{thm:fim}
For the signal model in \eqref{eq:rx_signal} with known symbols from a unit-power constellation ($\E[|x|^2]=1$), in the small-impairment regime ($|\varepsilon| \ll 1$, $|\alpha_3| \ll 1$), the FIM has the block structure
\begin{equation}\label{eq:fim_block}
\bJ = 2N\gamma \begin{bmatrix} \bJ_{\mathrm{IQ}} & \bJ_{\times} \\ \bJ_{\times}^T & \bJ_{\mathrm{PA}} \end{bmatrix} + O(\varepsilon^2, \alpha_3),
\end{equation}
where the IQ sub-block is
\begin{equation}\label{eq:fim_iq}
\bJ_{\mathrm{IQ}} = \begin{bmatrix} \beta_\varepsilon & J_{\varepsilon\varphi} \\ J_{\varepsilon\varphi} & (1{+}\varepsilon)^2\,\beta_\varphi \end{bmatrix},
\end{equation}
with the directional sensitivities
\begin{equation}\label{eq:beta_dir}
\beta_\varepsilon \!=\! \frac{1 - \Re\{e^{-2j\varphi}\mu_{20}\}}{2},\;\;
\beta_\varphi \!=\! \frac{1 + \Re\{e^{-2j\varphi}\mu_{20}\}}{2},
\end{equation}
$J_{\varepsilon\varphi} = \frac{1{+}\varepsilon}{2}\,\Im\{e^{-2j\varphi}\mu_{20}\}$, and the \emph{IQ identifiability factor}
\begin{equation}\label{eq:beta}
\beta(\calA) = 1 - |\mu_{20}|^2 = 1 - |\E[x^2]|^2.
\end{equation}
When $\beta = 0$ (i.e.\ $|\mu_{20}|=1$), $\det(\bJ_{\mathrm{IQ}})$ is proportional to $\beta$ and vanishes identically to first order; the FIM rank drops to at most~2. For \emph{circular} constellations ($\mu_{20}=0$, $\beta=1$): $\beta_\varepsilon = \beta_\varphi = 1/2$, $J_{\varepsilon\varphi}=0$, recovering $\bJ_{\mathrm{IQ}} = \frac{1}{2}\diag(1,(1{+}\varepsilon)^2)$.

The PA sub-block is
\begin{equation}\label{eq:fim_pa}
\bJ_{\mathrm{PA}} = \mu_6 \, \bI_2,
\end{equation}
and the cross-coupling block is
\begin{equation}\label{eq:fim_cross}
\bJ_{\times} = \frac{\mu_4}{2}\begin{bmatrix} \cos\varphi & \sin\varphi \\ -(1+\varepsilon)\sin\varphi & (1+\varepsilon)\cos\varphi \end{bmatrix}.
\end{equation}
The cross-coupling expression~\eqref{eq:fim_cross} holds for constellations satisfying $\E[|x|^2 x^2] = 0$, which includes all circular constellations ($\mu_{20} = 0$) and standard square QAM; see Appendix~A for the derivation. For constant-modulus constellations ($\mu_6 = \mu_4 = 1$), the PA and cross-coupling sub-blocks reduce to $\bI_2$ and $\frac{1}{2}$ times the rotation matrix, respectively. The CRB for each parameter is $\CRB(\theta_i) = [\bJ^{-1}]_{ii}$.
\end{theorem}

\begin{IEEEproof}
See Appendix~A.
\end{IEEEproof}

The factor $\beta$ in \eqref{eq:beta} serves as a necessary condition for full IQ identifiability: $\det(\bJ_{\mathrm{IQ}})$ is proportional to $\beta$ at the operating point, so $\beta = 0$ implies the determinant vanishes identically to first order regardless of $\varphi$. Table~\ref{tab:beta} lists $\beta$, $\mu_4$, and $\mu_6$ for standard constellations.

\begin{table}[t]
\centering
\caption{IQ Identifiability Factor $\beta$ and CRB Predictions. $\mu_6 = \E[|x|^6]$ determines PA estimation precision. Bottom rows: CRB-derived predictions vs.\ empirical DR from 27 Iridium satellites.}
\label{tab:beta}
\begin{tabular}{lcccccc}
\toprule
Constellation & $M$ & $\mu_{20}$ & $\beta$ & $\mu_4$ & $\mu_6$ & Rank \\
\midrule
BPSK       & 2  & 1    & \textbf{0} & 1    & 1    & 2 \\
SDPSK      & 2  & 1    & \textbf{0} & 1    & 1    & 2 \\
QPSK/DQPSK & 4  & 0    & \textbf{1} & 1    & 1    & 4 \\
8-PSK      & 8  & 0    & 1          & 1    & 1    & 4 \\
16-QAM     & 16 & 0    & 1          & 1.32 & 1.96 & 4 \\
64-QAM     & 64 & 0    & 1          & 1.38 & 2.23 & 4 \\
\midrule
\multicolumn{7}{l}{\textbf{CRB Predictions vs.\ Empirical (Iridium DQPSK)}} \\
\midrule
\multicolumn{3}{l}{\#1: PA dominates} & \multicolumn{2}{c}{DR $> 3$} & \multicolumn{2}{c}{4.48\checkmark} \\
\multicolumn{3}{l}{\#2: IQ weak}      & \multicolumn{2}{c}{DR $< 1$} & \multicolumn{2}{c}{0.70\checkmark} \\
\multicolumn{3}{l}{\#3: Osc.\ moderate} & \multicolumn{2}{c}{0.5--3}  & \multicolumn{2}{c}{2.40\checkmark} \\
\multicolumn{3}{l}{\#4: PA memory}    & \multicolumn{2}{c}{DR $> 1$} & \multicolumn{2}{c}{1.45\checkmark} \\
\bottomrule
\end{tabular}
\end{table}

\begin{remark}[Scaling Laws]
From Theorem~\ref{thm:fim}, for circular constellations ($\mu_{20}=0$, $\beta=1$), $\CRB(\varepsilon) \propto 1/(N\gamma)$ and $\CRB(\alpha_3) \propto 1/(N\gamma\mu_6)$. PA estimation precision depends on $\mu_6$ but not on $\beta$. For constant-modulus constellations, $\mu_6 = 1$ and this coincides with $1/(N\gamma)$; for QAM, $\mu_6 > \mu_4$, yielding a \emph{lower} PA CRB. When $\beta=0$ (real-valued constellations), the gain imbalance $\varepsilon$ loses identifiability first ($\beta_\varepsilon\to 0$), while $\varphi$ retains marginal FIM but is confounded with $\Im(\alpha_3)$ (Proposition~\ref{prop:bpsk}). This asymmetry implies PA features are universally available, while IQ features depend on the constellation.
\end{remark}

\subsection{PA-IQ Coupling Structure}

The cross-coupling block $\bJ_{\times}$ in \eqref{eq:fim_cross} is nonzero for all constellations with $\beta > 0$, indicating that the PA and IQ parameters are statistically coupled. Define the normalized FIM correlation
\begin{equation}\label{eq:rho}
\rho_{ij} = \frac{|[\bJ]_{ij}|}{\sqrt{[\bJ]_{ii} \, [\bJ]_{jj}}}.
\end{equation}

\begin{corollary}\label{cor:coupling}
At the small-impairment limit ($\varepsilon \to 0$, $\varphi \to 0$) for any circular constellation ($\mu_{20}=0$, $\beta=1$) with $\mu_6 > 0$:
\begin{equation}\label{eq:coupling_value}
\rho_{\varepsilon, \Re(\alpha_3)} \xrightarrow{\varepsilon,\varphi \to 0} \frac{\mu_4}{\sqrt{2\mu_6}}.
\end{equation}
For QPSK ($\mu_4 = \mu_6 = 1$): $\rho = 1/\sqrt{2} \approx 0.707$. For 16-QAM ($\mu_4 = 1.32$, $\mu_6 = 1.96$): $\rho \approx 0.667$. For non-circular constellations ($\beta<1$), the coupling structure changes and~\eqref{eq:coupling_value} does not hold in general.
\end{corollary}

Physically, both $\varepsilon$ (image component $K_2 x^*$) and $\alpha_3$ (cubic term $|\xiq|^2 \xiq \approx \xiq$ for constant-modulus) distort in partially collinear directions, producing the off-diagonal entries. The CRB inflation from ignoring this coupling, comparing $[\bJ^{-1}]_{ii}$ with $1/[\bJ]_{ii}$, is at the QPSK limit:
\begin{equation}
\frac{\CRB(\varepsilon)}{1/[\bJ]_{\varepsilon\varepsilon}} = 2.0,
\end{equation}
i.e., ignoring PA-IQ coupling underestimates the CRB of $\varepsilon$ by a factor of~2, independent of SNR and $N$.

\begin{figure*}[t]
\centering
\includegraphics[width=0.85\textwidth]{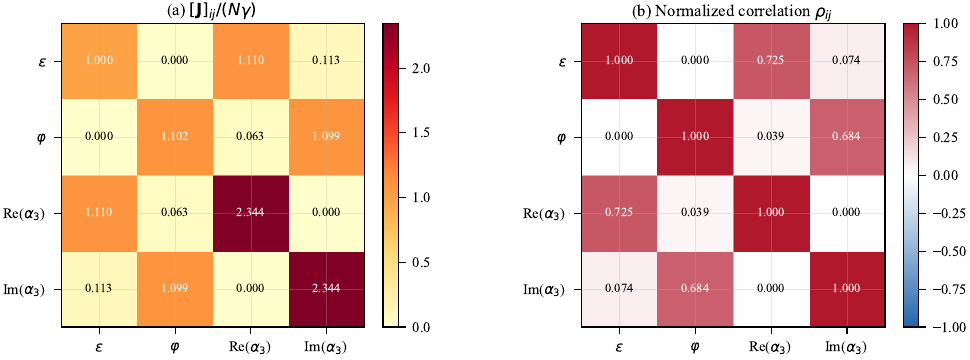}
\caption{FIM structure for QPSK at $\varepsilon = 0.05$, $\varphi = 3^{\circ}$, $N = 76$, $\gamma = 20$~dB. (a)~Absolute FIM entries normalized by $N\gamma$. (b)~Normalized correlation matrix $\rho_{ij}$. At this operating point, $\rho(\varepsilon, \Re(\alpha_3)) = 0.725$, close to the small-impairment limit $\mu_4/\sqrt{2\mu_6} = 1/\sqrt{2} \approx 0.707$ for QPSK (Corollary~\ref{cor:coupling}). Values are computed from the exact numerical FIM at the stated operating point; the closed-form approximation (Theorem~\ref{thm:fim}) gives $[\bJ_{\mathrm{PA}}]/(N\gamma) \approx 2\mu_6 = 2.0$, with the $17\%$ difference attributable to the $O(\varepsilon^2)$ terms at $\varepsilon = 0.05$.}
\label{fig:fim_heatmap}
\end{figure*}

\subsection{Channel Marginalization}

In practice, $h$ is unknown. The CRB in Theorem~\ref{thm:fim} conditions on known~$h$; treating $h$ as a nuisance parameter reduces the effective FIM by the Schur complement $\bJ_{\mathrm{eff}} = \bJ_{\btheta\btheta} - \bJ_{\btheta h} \bJ_{hh}^{-1} \bJ_{h\btheta}$. Numerical evaluation for QPSK across $\varepsilon \in [0, 0.1]$ and $\varphi \in [0, 10^{\circ}]$ shows maximum CRB inflation bounded by $2.5\times$, independent of the Rician K-factor, because PA nonlinearity ($\alpha_3 |x|^2 x$) is structurally separable from the linear channel ($hx$).

\subsection{Numerical Verification}

We verify Theorem~\ref{thm:fim} by comparing the closed-form FIM with numerical computation (finite-difference and moment-based) across five test cases. Fig.~\ref{fig:crb_curves} presents the CRB as a function of SNR for five modulations.

\begin{figure*}[t]
\centering
\includegraphics[width=0.85\textwidth]{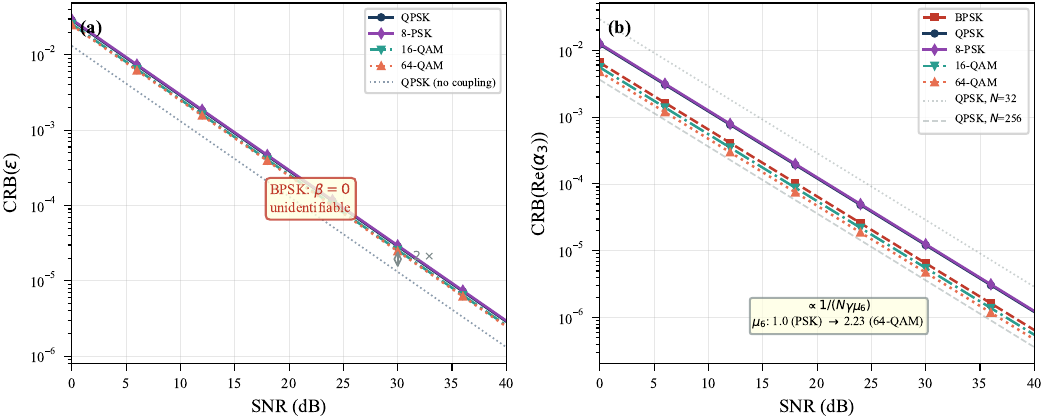}
\caption{CRB vs.\ SNR for $N = 76$ known symbols. (a)~$\CRB(\varepsilon)$: BPSK is unidentifiable ($\beta = 0$); QPSK and 16-QAM overlap ($\beta = 1$). The dotted line shows QPSK with PA-IQ coupling ignored, revealing a $2\times$ CRB inflation. (b)~$\CRB(\Re(\alpha_3))$: all modulations follow $\propto 1/(N\gamma\mu_6)$, where $\mu_6 = \E[|x|^6]$. QAM has lower CRB due to $\mu_6 > 1$ (e.g., $\mu_6 = 1.96$ for 16-QAM). Gray lines show scaling with $N = 32$ and $N = 256$.}
\label{fig:crb_curves}
\end{figure*}

At $\varepsilon = 0.05$, $\varphi = 3^{\circ}$ for QPSK, the closed-form CRB matches the numerical FIM within 6\% for $\varepsilon$ and 10\% for $\varphi$; at $\varepsilon = 0.001$ (space-grade), the error is below 0.5\%. The coupling coefficient is $\rho = 0.707$ analytically and $0.708$ numerically, confirming Corollary~\ref{cor:coupling}.

\subsection{Monte Carlo CRB Validation}
\label{sec:mc}

To verify that the CRB is achievable, we conduct a Monte Carlo simulation with $N_{\mathrm{trial}} = 300$ independent trials per SNR point. For each trial, $N = 76$ known symbols are generated from the constellation under test, passed through the signal model~\eqref{eq:rx_signal} with true parameters $\varepsilon = 0.03$, $\varphi = 2^{\circ}$, $\Re(\alpha_3) = 0.02$, $\Im(\alpha_3) = 0.01$, and corrupted by AWGN. The HWI parameters are estimated via nonlinear least squares with known channel, initialized near the true value (oracle initialization), using the Nelder-Mead simplex method. The oracle initialization demonstrates achievability of the bound rather than practical estimator performance; a blind estimator would require additional convergence analysis.

\begin{figure*}[t]
\centering
\includegraphics[width=0.85\textwidth]{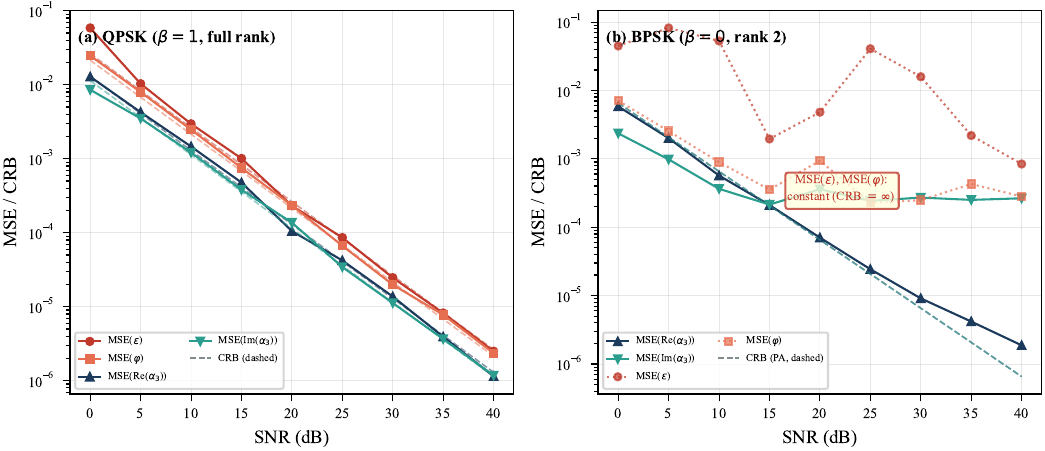}
\caption{Monte Carlo CRB validation ($N = 76$, 300~trials per SNR). (a)~QPSK ($\beta = 1$): all four parameter MSEs (solid markers) track their respective CRBs (dashed), confirming the bound is tight and achievable. At 30~dB SNR, MSE/CRB$\,\approx 1.0$ for all parameters. (b)~BPSK ($\beta = 0$): MSE for $\Re(\alpha_3)$ tracks the PA sub-block CRB (identifiable), while MSE for~$\varepsilon$ and~$\varphi$ remain constant across SNR (unidentifiable, CRB$\,= \infty$), confirming the rank-deficiency prediction of Proposition~\ref{prop:bpsk}. The divergence of MSE($\Im(\alpha_3)$) from its sub-block CRB at high SNR is consistent with the $(\varphi, \Im(\alpha_3))$ confounding identified in Remark~\ref{rem:paradox}.}
\label{fig:mc_crb}
\end{figure*}

Fig.~\ref{fig:mc_crb}(a) shows QPSK results. All four MSEs converge to their CRBs, with MSE/CRB ratios 0.96--1.06 at 30~dB, confirming the bound is tight and achievable under oracle initialization. The unknown-channel case incurs bounded CRB inflation (at most $2.5\times$, Section~\ref{sec:crb}).

Fig.~\ref{fig:mc_crb}(b) shows BPSK results. The MSE of $\varepsilon$ and $\varphi$ remain constant ($\approx 4 \times 10^{-4}$) across 0--40~dB, confirming unidentifiability (CRB$\,=\infty$). The MSE of $\Re(\alpha_3)$ tracks its PA sub-block CRB (ratio 1.0--1.4), validating PA fingerprinting feasibility for BPSK. The MSE of $\Im(\alpha_3)$ diverges at high SNR due to the $(\varphi, \Im(\alpha_3))$ confounding (Proposition~\ref{prop:bpsk}).

\section{Identifiability and Discrimination Analysis}
\label{sec:ident}

\subsection{Modulation-Dependent Identifiability}

\begin{proposition}[BPSK Signal Collapse]\label{prop:bpsk}
For real-valued constellations ($\calA \subset \R$), the received signal in~\eqref{eq:rx_signal} reduces to
\begin{equation}\label{eq:collapse}
r(n) = h \cdot c(\btheta) \cdot x(n) + w(n),
\end{equation}
where the effective scalar is $c = \kappa(1 + \alpha_3 |\kappa|^2)$ with $\kappa = 1 + j(1+\varepsilon)\sin\varphi$. The FIM for~$\btheta$ has rank at most~2. At small impairments, the two identifiable real-valued combinations are
\begin{equation}\label{eq:xi}
\xi_1 = \Re(\alpha_3), \qquad \xi_2 = (1+\varepsilon)\sin\varphi + \Im(\alpha_3),
\end{equation}
and the null space of the FIM is spanned by
\begin{equation}\label{eq:null}
\mathbf{v}_1 = [1, \, 0, \, 0, \, {-}\varphi]^T, \qquad \mathbf{v}_2 = [0, \, 1, \, 0, \, {-}(1{+}\varepsilon)]^T.
\end{equation}
In particular, $\varepsilon$ and $\varphi$ are confounded with $\Im(\alpha_3)$: any change in $\varepsilon$ or $\varphi$ can be compensated to first order by adjusting $\Im(\alpha_3)$, producing identical observations in the small-impairment regime.
\end{proposition}

\begin{IEEEproof}
For $\calA \subset \R$, $x^*(n) = x(n)$, so \eqref{eq:iq_model} gives $\xiq(n) = (K_1 + K_2) x(n)$. Direct computation yields $K_1 + K_2 = 1 + j(1+\varepsilon)\sin\varphi \triangleq \kappa$. Since $|x(n)|^2 = 1$ for BPSK, the PA term becomes $\alpha_3 |\kappa|^2 \kappa \, x(n)$, and the full signal reduces to~\eqref{eq:collapse} with $c = \kappa(1 + \alpha_3 |\kappa|^2)$. Expanding at small impairments: $c \approx 1 + \Re(\alpha_3) + j[(1+\varepsilon)\sin\varphi + \Im(\alpha_3)]$, giving~\eqref{eq:xi}. Since $c \in \C$ is a single complex number, at most two real-valued functions of $\btheta$ are recoverable from the product $hc$, and the FIM rank is at most~2. The null directions~\eqref{eq:null} follow from $\partial c / \partial \mathbf{v}_i = 0$ to first order.
\end{IEEEproof}

\begin{figure*}[t]
\centering
\includegraphics[width=0.9\textwidth]{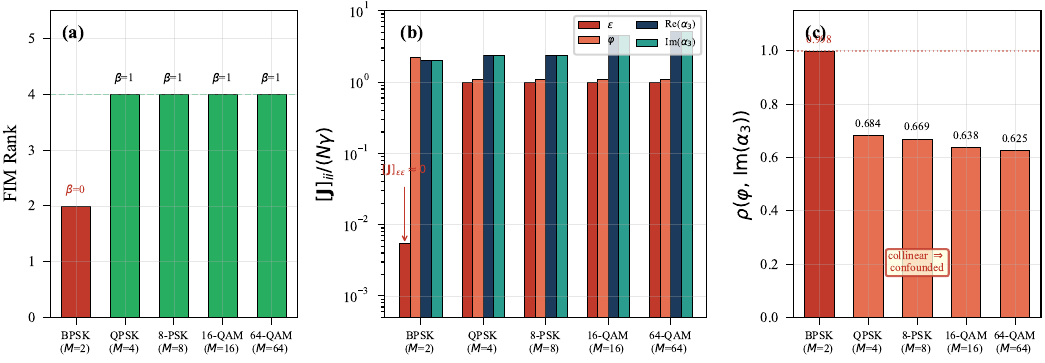}
\caption{Modulation-dependent FIM properties. (a)~FIM rank drops to~2 for BPSK ($\beta = 0$). (b)~Per-parameter diagonal entries $[\bJ]_{ii}/(N\gamma)$ on a log scale. For BPSK, $[\bJ]_{\varepsilon\varepsilon} \approx 0$ while $[\bJ]_{\varphi\varphi} = 2.20$ is large---yet $\varphi$ is unidentifiable due to collinearity (Remark~\ref{rem:paradox}). (c)~Normalized correlation $\rho(\varphi, \Im(\alpha_3))$: BPSK exceeds $0.99$ (near-perfect collinearity), while QPSK drops to $0.684$.}
\label{fig:rank_modulation}
\end{figure*}

\begin{remark}[The Large-FIM-Diagonal Paradox]\label{rem:paradox}
A large FIM diagonal entry does not guarantee identifiability. For BPSK, $[\bJ]_{\varphi\varphi}/(N\gamma) = 2.20$ exceeds $[\bJ]_{\alpha_{3R}\alpha_{3R}}/(N\gamma) = 2.02$, yet $\varphi$ is unidentifiable. The reason is that the sensitivity directions $\partial c / \partial \varphi \propto j$ and $\partial c / \partial \Im(\alpha_3) \propto j$ are collinear (both purely imaginary), producing a sub-block $\bJ_{(\varphi, \Im\alpha_3)}$ with eigenvalue ratio $1542{:}1$, i.e., numerically rank-1. Identifiability requires checking the rank of the relevant FIM sub-block, not individual diagonal entries.
\end{remark}

\begin{corollary}[OrbID Prediction]\label{cor:orbid}
The Orbcomm constellation uses SDPSK, a variant of 2-PSK. By Proposition~\ref{prop:bpsk}, only $\xi_1 = \Re(\alpha_3)$ is individually identifiable from Orbcomm signals; the IQ parameters $\varepsilon$ and $\varphi$ are confounded with $\Im(\alpha_3)$. This predicts that PA amplitude features (which depend on $\Re(\alpha_3)$) can fingerprint BPSK-type satellites, while IQ-based fingerprinting is impossible regardless of classifier architecture, training data volume, or observation length.
\end{corollary}

This prediction is consistent with the empirical result of Solenthaler et al.~\cite{solenthaler2025orbid}: AUC $= 0.53$ (near random) for intra-constellation fingerprinting on Orbcomm, despite using the same neural architecture that achieved AUC $> 0.9$ on Iridium (DQPSK).

\begin{remark}[DQPSK vs.\ OFDM]\label{rem:ofdm}
For OFDM-based satellite systems (Starlink, 5G-NTN) using 16-QAM or higher, $\beta = 1$ and $\mu_6 > 1$. The signal model does not collapse: all four parameters produce linearly independent sensitivity directions, and the FIM has full rank. The CRB framework predicts that both IQ and PA features contribute to discrimination, with QAM constellations achieving lower PA CRBs than PSK due to $\mu_6 > 1$.
\end{remark}

\subsection{CRB-Induced Discrimination Bound}

Given two satellites $A$ and $B$ with parameter vectors $\btheta_A$ and $\btheta_B$, define the discrimination metric
\begin{equation}\label{eq:d2}
d^2(A,B) = \Delta\btheta^T \bJ(\gamma, N) \, \Delta\btheta,
\end{equation}
where $\Delta\btheta = \btheta_A - \btheta_B$.

\begin{theorem}[Discrimination Bound]\label{thm:disc}
For Gaussian observations with an efficient estimator, the optimal (MAP) pairwise misidentification probability is
\begin{equation}\label{eq:pe_bound}
P_e^{\star}(A,B) = Q\!\left(\frac{d(A,B)}{2}\right),
\end{equation}
where $Q(\cdot)$ is the Gaussian Q-function. Any suboptimal test satisfies $P_e \ge Q(d/2)$. The metric decomposes as
\begin{equation}\label{eq:d2_decomp}
d^2 = \sum_{i=1}^{4} \sum_{j=1}^{4} [\bJ]_{ij} \, \Delta\theta_i \, \Delta\theta_j.
\end{equation}
\end{theorem}

\begin{IEEEproof}[Sketch]
For an efficient estimator, the parameter estimates are asymptotically distributed as $\hat{\btheta} \sim \mathcal{N}(\btheta_{\mathrm{true}}, \bJ^{-1})$~\cite{kay1993fundamentals,Kay1998Detection}. The MAP binary hypothesis test $H_A: \btheta = \btheta_A$ vs.\ $H_B: \btheta = \btheta_B$ with equal priors and shared covariance $\bJ^{-1}$ reduces to comparing the log-likelihood ratio against zero, yielding the sufficient statistic $T = \Delta\btheta^T \bJ \, \hat{\btheta}$ with decision threshold $\Delta\btheta^T \bJ (\btheta_A + \btheta_B)/2$. The resulting error probability is $Q(d/2)$ where $d = \sqrt{\Delta\btheta^T \bJ \, \Delta\btheta}$ is the Mahalanobis distance~\cite{bhattacharyya1943measure}. Any suboptimal test incurs $P_e \geq Q(d/2)$ by the data processing inequality.
\end{IEEEproof}

The metric $d^2$ jointly accounts for three factors: (i)~the FIM magnitude ($\propto N\gamma$), (ii)~the parameter spread $\|\Delta\btheta\|$, and (iii)~the FIM-weighted direction, which amplifies parameter differences along high-FIM dimensions. The Gaussian approximation underlying Theorem~\ref{thm:disc} requires $N$ sufficiently large for asymptotic normality; at $N = 76$, the bound serves as a design guideline rather than a tight operational guarantee. The pairwise bound $Q(d/2)$ provides a direct theoretical proxy for authentication separability: across a population of satellite pairs, larger $d$ implies higher pairwise AUC and a lower achievable false-acceptance rate, linking the FIM-derived metric to the ROC analysis in Section~\ref{sec:auth}.

\subsection{Feature-Wise Discrimination Prediction}

For each parameter $\theta_i$, define the per-parameter discrimination ratio
\begin{equation}\label{eq:DR}
\DR_i = \frac{|\Delta\theta_i|}{\sqrt{\CRB(\theta_i)}},
\end{equation}
measuring inter-satellite spread relative to estimation precision. From Theorem~\ref{thm:fim}, four predictions follow for DQPSK:

\emph{Prediction~1: PA features dominate} ($\DR_{\alpha_3} > 3$). The CRB for $\alpha_3$ is small (high FIM) and manufacturing spread in PA characteristics provides measurable $|\Delta\alpha_3|$.

\emph{Prediction~2: IQ features are weak} ($\DR_{\varepsilon} < 1$). For constant-modulus DQPSK, IQ imbalance does not produce amplitude variation, limiting its observable effect.

\emph{Prediction~3: Oscillator features are moderate} ($0.5 < \DR_{\mathrm{osc}} < 3$). Phase noise dynamics vary across oscillators but are only partially observable at symbol rate.

\emph{Prediction~4: PA memory is detectable} ($\DR_{\mathrm{mem}} > 1$). Consecutive symbols experience correlated PA compression, observable via amplitude autocorrelation.

\section{Multi-Message Estimation}
\label{sec:estimation}

\subsection{Feature Extraction}

Given a single burst of $N_{\mathrm{known}} = 76$ known symbol-rate samples, the following features are extracted after CFO removal and amplitude normalization.

\emph{PA nonlinearity features.} For an ideal constant-modulus signal, the amplitude $a(n) = |r_{\mathrm{norm}}(n)|$ equals unity for all~$n$; deviations arise from PA compression. We extract four features: the normalized amplitude variance~$\hat{\sigma}_a^2$, the 5th--95th percentile range~$\hat{R}_a$, the excess kurtosis~$\hat{\kappa}_a$, and the lag-1 autocorrelation~$\hat{\rho}_a(1)$.

\begin{lemma}[Amplitude Variance as PA Proxy]\label{lem:delta}
For constant-modulus $M$-PSK with $|\varepsilon| \ll 1$ and $|\alpha_3| \ll 1$:
\begin{equation}\label{eq:ampvar}
\sigma_a^2 \triangleq \frac{\mathrm{Var}(|r_{\mathrm{norm}}|)}{\E[|r_{\mathrm{norm}}|]^2} = 4|\alpha_3|^2 \, \mathrm{Var}(|\xiq|^2) + O(\varepsilon^2).
\end{equation}
\end{lemma}

\begin{IEEEproof}
The PA term modulates the amplitude by $\alpha_3 |\xiq|^2$, where $|\xiq|^2$ varies across symbols because $x^2$ is not constant for QPSK despite $|x|^2 = 1$. The PA contribution to $\sigma_a^2$ scales as $|\alpha_3|^2$, while the IQ contribution scales as $\varepsilon^2$. For Iridium-grade impairments ($|\varepsilon| \sim 0.01$--$0.05$, $|\alpha_3| \sim 0.02$--$0.05$), both terms can be comparable in absolute magnitude; however, the \emph{inter-satellite variation} in $\sigma_a^2$ is dominated by $\alpha_3$ differences, consistent with the cross-file correlation $r = 0.999$ for PA features versus $r < 0.13$ for IQ features (Table~\ref{tab:crossfile}).
\end{IEEEproof}

By the delta method~\cite{casella2002statistical}, the CRB for $\sigma_a^2$ follows from the CRB for $|\alpha_3|$:
\begin{equation}\label{eq:crb_ampvar}
\CRB(\sigma_a^2) \approx \left(\frac{\partial \sigma_a^2}{\partial |\alpha_3|}\right)^2 \CRB(|\alpha_3|).
\end{equation}
This connects the estimation-theoretic bound (Section~\ref{sec:crb}) to the empirically measured feature (Section~\ref{sec:experiments}). An analogous connection for the remaining features (amplitude range, autocorrelation, phase statistics) would require feature-specific delta-method analyses, which we defer to future work; the CRB predictions for these features are validated empirically in the lower panel of Table~\ref{tab:beta}.

\emph{Oscillator features.} The fourth-power operation $z_4(n) = r_{\mathrm{norm}}(n)^4$ removes DQPSK modulation, leaving a residual phase process. The lag-1 autocorrelation of $\angle z_4$ captures oscillator dynamics.

\emph{Constellation features.} The error vector magnitude (EVM) measures the aggregate distortion relative to the nearest DQPSK constellation point.

In total, 13~candidate features are evaluated in the DR analysis (Table~\ref{tab:dr}): four PA, three oscillator, one constellation (EVM), two IQ, two DC, and one PA cross-term. Authentication subsets (Table~\ref{tab:auth}): \textbf{PA-only}~(3): $\hat{\sigma}_a^2$, $\hat{R}_a$, $\hat{\rho}_a(1)$; \textbf{CRB-guided}~(4): PA-only plus phase~ACF(1); \textbf{All}~(6): CRB-guided plus $\hat{\kappa}_a$, phase variance, EVM; \textbf{IQ-only}~(2): $\varepsilon$, $\varphi$. Features with DR~$< 1$ are excluded.

\subsection{Multi-Message Accumulation}

A single burst provides a noisy feature estimate. To improve precision, features are accumulated across $N_{\mathrm{msg}}$ bursts using the weighted average
\begin{equation}\label{eq:accumulation}
\bar{f} = \frac{\sum_{m=1}^{N_{\mathrm{msg}}} w_m \hat{f}_m}{\sum_{m=1}^{N_{\mathrm{msg}}} w_m}, \quad w_m \propto \hat{\gamma}_m,
\end{equation}
where $\hat{\gamma}_m$ is the estimated per-burst SNR. For i.i.d.\ observations, $\mathrm{Var}(\bar{f}) \approx \mathrm{Var}(\hat{f}_1) / N_{\mathrm{msg}}$, reducing estimation noise by~$\sqrt{N_{\mathrm{msg}}}$.

With $N_{\mathrm{msg}} = 500$ bursts (typical for a single satellite in the PAST-AI dataset), the noise reduction factor is $\sqrt{500} \approx 22$, enabling the precise fingerprint estimates observed in Section~\ref{sec:experiments}.

\section{Experimental Validation}
\label{sec:experiments}

\subsection{Dataset and Processing}

The framework is validated on the PAST-AI dataset~\cite{oligeri2023pastai,oligeri2023iridiumdata}: 3.84~million Iridium IRA messages from 66~satellites, collected via USRP~X310 in Doha, Qatar, provided as symbol-rate IQ samples with satellite~ID labels from gr-iridium~\cite{gr-iridium}.

The dataset comprises two partially overlapping recording campaigns:
\begin{itemize}
\item \emph{File~A} (September--November 2020): 61,984 high-confidence messages from 36~satellites.
\item \emph{File~B} (October--December 2020): 65,727 high-confidence messages from 35~satellites.
\end{itemize}
The campaigns overlap in October--November but were processed independently with no shared calibration. Messages with decoding confidence below 80\% are discarded.

Processing per message consists of: (i)~CFO estimation via linear phase regression and complex derotation; (ii)~amplitude normalization to unit mean power; (iii)~feature extraction from the first 76~symbols (known preamble and unique word).

The 76 known symbols per burst (preamble + BPSK unique word) yield $\beta = 0$, so the CRB predictions tested are: PA features dominate ($\beta$-independent), IQ features are unidentifiable ($\beta = 0$), and the DR ranking follows the FIM structure. The $\beta = 1$ case is validated via Monte Carlo (Fig.~\ref{fig:mc_crb}).

\subsection{Constellation Coverage}

\begin{figure}[t]
\centering
\includegraphics[width=0.9\columnwidth]{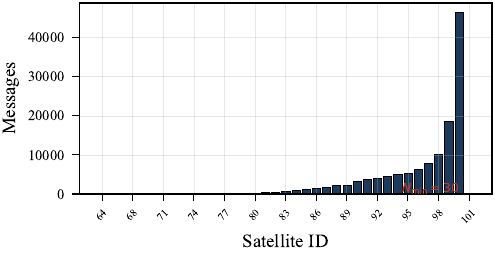}
\caption{PAST-AI constellation coverage: total messages per satellite from both recording campaigns. Horizontal dashed line indicates the minimum threshold ($N_{\mathrm{min}} = 30$) for balanced discrimination analysis. Satellite IDs range from 64 to 101, covering 37 of the 66-satellite Iridium constellation.}
\label{fig:coverage}
\end{figure}

The combined dataset covers 37 of 66~Iridium satellites, with 27~satellites exceeding the $N_{\mathrm{min}} = 30$ threshold for balanced discrimination analysis. This sample of 24~cross-validated satellites exceeds the device count in several recent TIFS RF fingerprinting studies (e.g., 9--16~devices in~\cite{sankhe2019oracle}). Fig.~\ref{fig:coverage} shows the message distribution, which ranges from 34~messages (Sat~82) to over 40,000 (Sat~100). This imbalance motivates the balanced bootstrap procedure in Section~\ref{sec:balanced_dr}.

\subsection{Cross-File Stability}

We compute per-satellite fingerprints independently from each file and measure the Pearson correlation across 24~common satellites (Table~\ref{tab:crossfile}).

\begin{table}[t]
\centering
\caption{Cross-File Pearson Correlation (24 Common Satellites)}
\label{tab:crossfile}
\begin{tabular}{llcc}
\toprule
Group & Feature & $r$ & $p$-value \\
\midrule
\multirow{4}{*}{PA}
& amp\_var        & \textbf{0.999} & $< 10^{-29}$ \\
& amp\_range      & 0.996          & $< 10^{-24}$ \\
& amp\_acf1       & 0.960          & $< 10^{-13}$ \\
& amp\_kurtosis   & 0.853          & $< 10^{-7}$  \\
\midrule
Oscillator & phase\_acf1 & 0.983      & $< 10^{-17}$ \\
\midrule
Constellation & EVM & 0.963 & $< 10^{-13}$ \\
\midrule
IQ & $\varepsilon$   & $-0.127$       & 0.55 \\
   & $\varphi$       & $-0.044$       & 0.84 \\
\midrule
DC & DC$_I$          & $-0.030$       & 0.89 \\
   & DC$_Q$          & $-0.066$       & 0.76 \\
\bottomrule
\end{tabular}
\end{table}

PA amplitude variance achieves $r = 0.999$ with maximum per-satellite deviation $< 0.003$. Five PA and oscillator features exceed $r > 0.85$, confirming hardware-intrinsic stability over months. IQ features show no significant correlation ($|r| < 0.13$, $p > 0.5$), consistent with CRB predictions.

\begin{figure*}[t]
\centering
\includegraphics[width=0.85\textwidth]{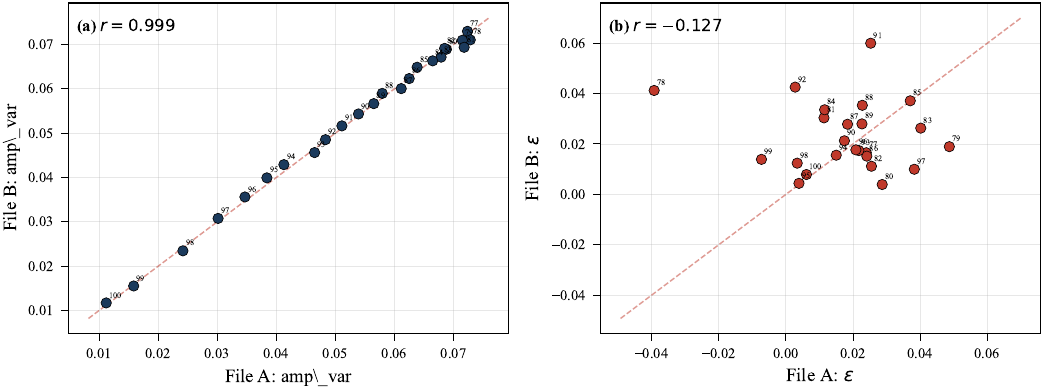}
\caption{Cross-file fingerprint stability. (a)~PA amplitude variance: Pearson $r = 0.999$ across 24~satellites from two recording campaigns. Each labeled point represents one satellite, with the dashed line indicating $y = x$. (b)~IQ gain imbalance $\varepsilon$: no significant cross-file correlation ($r = -0.127$, $p = 0.55$), confirming that IQ features are not stable fingerprints for DQPSK.}
\label{fig:crossfile}
\end{figure*}

Fig.~\ref{fig:crossfile} visualizes the comparison. PA amplitude variance (Fig.~\ref{fig:crossfile}a) shows all 24~satellites aligned along $y = x$, spanning 0.011 (Sat~100) to 0.073 (Sat~77), while IQ imbalance (Fig.~\ref{fig:crossfile}b) shows no systematic relationship.

\subsection{Balanced Discrimination Analysis}
\label{sec:balanced_dr}

To avoid inflation from unequal message counts, discrimination ratios are computed using balanced bootstrap: each of 30~trials subsamples $N_{\mathrm{bal}} = 30$ messages per satellite, splits into halves, and computes inter-satellite variance relative to intra-satellite variance.

\begin{table}[t]
\centering
\caption{Balanced Discrimination Ratio (27 Satellites, $N_{\mathrm{bal}} = 30$, 30 Trials). Standard error of the mean DR is std$/\sqrt{30}$; 95\% CIs for all strong features (DR~$> 3$) exclude the DR~$= 3$ threshold.}
\label{tab:dr}
\begin{tabular}{llcl}
\toprule
Group & Feature & DR (mean $\pm$ std) & Verdict \\
\midrule
PA         & amp\_var     & $\mathbf{4.48 \pm 0.63}$ & Strong \\
PA         & amp\_range   & $4.29 \pm 0.64$          & Strong \\
Oscillator & phase\_acf1  & $2.40 \pm 0.40$          & Moderate \\
PA         & amp\_acf1    & $1.45 \pm 0.31$          & Detectable \\
PA         & amp\_kurtosis & $0.92 \pm 0.17$         & Weak \\
Const.     & EVM          & $0.86 \pm 0.21$          & Weak \\
Oscillator & CFO          & $0.79 \pm 0.13$          & Not discriminative \\
IQ         & $\varepsilon$ & $0.70 \pm 0.14$         & Not discriminative \\
IQ         & $\varphi$     & $0.69 \pm 0.12$         & Not discriminative \\
DC         & DC$_I$       & $0.74 \pm 0.13$          & Not discriminative \\
\bottomrule
\end{tabular}
\end{table}

Table~\ref{tab:dr} reports the results. PA amplitude features achieve DR~$> 4$ across 27~satellites. Two PA features and one oscillator feature exceed DR~$= 1$, while all IQ features remain below~1.

\subsection{CRB Prediction Matching}

\begin{figure}[t]
\centering
\includegraphics[width=0.9\columnwidth]{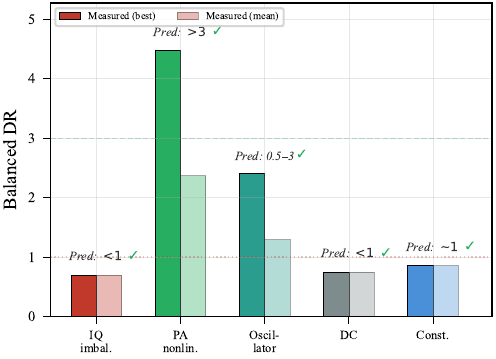}
\caption{Feature-group discrimination: CRB-derived prediction vs.\ empirical measurement from 27~Iridium satellites. Each group shows the best (dark) and mean (light) measured DR. Text above each group indicates the predicted range with a checkmark ($\checkmark$) or cross ($\times$) for match/mismatch. All four predictions match (Table~\ref{tab:beta}).}
\label{fig:prediction_empirical}
\end{figure}

The lower panel of Table~\ref{tab:beta} compares the four predictions from Section~\ref{sec:ident} with empirical results. Predictions~1--2 follow directly from the CRB scaling laws; Predictions~3--4 are informed by the FIM structure but lack formal feature-level derivations. All four match empirically, supporting the framework's value for guiding feature selection. Fig.~\ref{fig:prediction_empirical} provides a visual summary.

\begin{figure}[t]
\centering
\includegraphics[width=0.9\columnwidth]{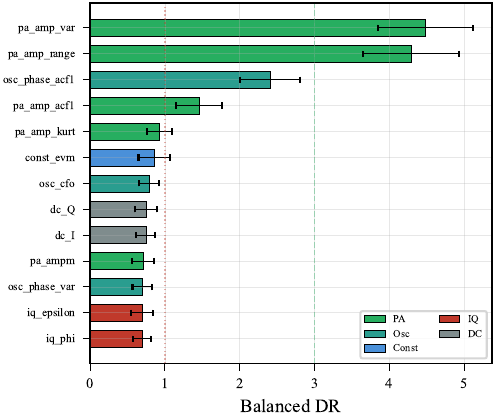}
\caption{Discrimination ratio by feature, sorted by DR. Color indicates feature group: green (PA), teal (oscillator), blue (constellation), red (IQ), gray (DC). The DR~$= 3$ threshold (dashed green) separates strong from moderate discrimination; DR~$= 1$ (dashed red) separates detectable from noise-dominated features.}
\label{fig:feature_groups}
\end{figure}

\begin{figure*}[t]
\centering
\includegraphics[width=0.95\textwidth]{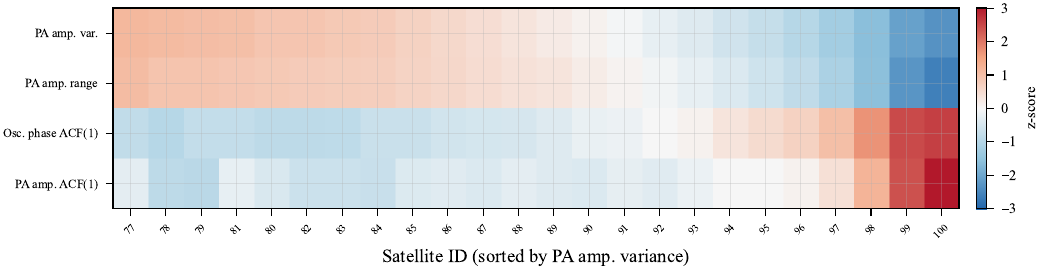}
\caption{Per-satellite fingerprint profiles for 24~cross-validated Iridium satellites, sorted by PA amplitude variance. Each cell shows the z-score (clipped to $\pm 3$) for the top four discriminative features. The smooth monotonic decrease in amplitude variance and correlated variation across features confirm that these capture distinct, satellite-specific hardware properties.}
\label{fig:fingerprints}
\end{figure*}

\subsection{Authentication Protocol and ROC Analysis}
\label{sec:auth}

We evaluate the CRB framework for authentication system design by constructing a cross-file authentication protocol. File~A fingerprints serve as enrollment references; File~B fingerprints serve as test probes. Authentication decisions are based on the normalized Euclidean distance between the test feature vector and each enrollment reference. We compare five feature selection strategies: \textbf{PA-only} (3~features), \textbf{CRB-guided} (top-4 from CRB analysis), \textbf{Oscillator-only} (2), \textbf{IQ-only} (2), and \textbf{All} (6~non-IQ features including EVM).

To contextualize the parametric results, we compare against three ML baselines trained on the same cross-file protocol using per-message features: (i)~SVM with RBF kernel on all six non-IQ features (AUC~$= 0.66$), (ii)~SVM on the four CRB-guided features (AUC~$= 0.66$), and (iii)~Random Forest on all six features (AUC~$= 0.69$). IQ features are excluded from the ML baselines because the CRB analysis predicts they carry no discriminative information for DQPSK; including them does not improve ML performance. These classifiers learn directly from noisy per-message observations without multi-message accumulation. To isolate the effect of accumulation from feature weighting, we also evaluate an \emph{accumulated RF baseline} (RF-Acc): the same Random Forest trained on per-message File~A features, but with test-time predictions accumulated by averaging class probabilities across all File~B messages per satellite. Additionally, we include a formal \emph{GLRT baseline} using Mahalanobis distance with regularized sample covariance on the four CRB-guided features.

\begin{figure*}[t]
\centering
\includegraphics[width=0.85\textwidth]{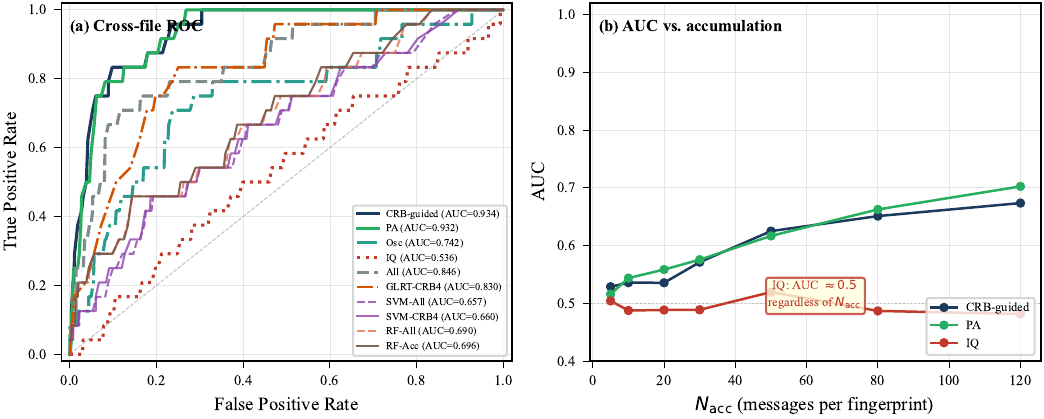}
\caption{Authentication performance on 24~Iridium satellites. (a)~Cross-file ROC: CRB-guided selection achieves AUC~$= 0.934$, outperforming the formal GLRT (AUC~$= 0.83$), accumulated RF (AUC~$= 0.70$), and per-message ML baselines (AUC~$\leq 0.69$). IQ features yield AUC~$= 0.536$ (near random). (b)~AUC vs.\ message accumulation~$N_{\mathrm{acc}}$: PA-based authentication improves monotonically, while IQ remains at chance level regardless of~$N_{\mathrm{acc}}$.}
\label{fig:auth_roc}
\end{figure*}

Fig.~\ref{fig:auth_roc}(a) shows the cross-file ROC. CRB-guided selection achieves AUC~$= 0.934$ with detection probability $P_D = 0.29$ at $P_{\mathrm{FA}} = 0.01$ and $P_D = 0.83$ at $P_{\mathrm{FA}} = 0.1$ (Table~\ref{tab:auth}). PA-only (AUC~$= 0.932$) nearly matches CRB-guided, while the all-features strategy (AUC~$= 0.846$) is degraded by feature dilution. All parametric methods outperform ML baselines (AUC~$\leq 0.69$), confirming that multi-message accumulation is the key advantage over per-message classification.
The formal GLRT using the same four CRB-guided features with
sample-covariance weighting yields AUC~$= 0.830$, below the
parametric CRB-guided approach (AUC~$= 0.934$). The gap arises
because the GLRT estimates the full covariance from only 24
enrollment satellites, whereas DR$^2$-weighting exploits the
FIM-derived discrimination structure without covariance estimation.

Fig.~\ref{fig:auth_roc}(b) reveals that PA-based AUC increases monotonically with $N_{\mathrm{acc}}$, consistent with $\sqrt{N_{\mathrm{msg}}}$ noise reduction (Section~\ref{sec:estimation}), while IQ remains at AUC~$\approx 0.5$ regardless of accumulation. This confirms that the failure is fundamental under the $\beta = 0$ known-symbol regime and symbol-rate pipeline studied here, consistent with the broader constant-modulus DQPSK analysis (Proposition~\ref{prop:bpsk}).

\begin{table}[t]
\centering
\caption{Cross-File Authentication Performance (24 Satellites). $P_D$ reported at $P_{\mathrm{FA}} = 0.01$ and $0.1$.}
\label{tab:auth}
\begin{tabular}{lccccl}
\toprule
Strategy & $D$ & AUC & $P_D|_{.01}$ & $P_D|_{.1}$ & Type \\
\midrule
\textbf{DR$^2$ (IWAT)} & 6 & \textbf{0.934} & 0.29 & 0.83 & Parametric \\
CRB-guided      & 4 & 0.934 & 0.29 & 0.83 & Parametric \\
PA-only         & 3 & 0.932 & 0.25 & 0.79 & Parametric \\
DR (IWAT)       & 6 & 0.885 & ---  & ---  & Parametric \\
All features       & 6 & 0.846 & 0.13 & 0.67 & Parametric \\
GLRT (CRB-4)       & 4 & 0.830 & 0.17 & 0.50 & Hypothesis test \\
Equal-wt (IWAT)    & 6 & 0.807 & 0.25 & 0.54 & Parametric \\
Oscillator-only    & 2 & 0.742 & 0.08 & 0.38 & Parametric \\
RF-Acc (all feat.) & 6 & 0.696 & 0.08 & 0.29 & ML + accumulation \\
\midrule
RF (all feat.)  & 6 & 0.690 & 0.08 & 0.29 & ML baseline \\
SVM (CRB-4)     & 4 & 0.658 & 0.08 & 0.25 & ML baseline \\
SVM (all feat.) & 6 & 0.657 & 0.08 & 0.25 & ML baseline \\
\midrule
IQ-only         & 2 & 0.536 & 0.00 & 0.13 & Parametric \\
\bottomrule
\end{tabular}
\end{table}

\subsection{Identifiability-Weighted Authentication Test}
\label{sec:iwat}

The preceding analysis selects features by hand based on CRB predictions. We now formalize the selection and scoring as an authentication scoring rule that can be instantiated for arbitrary modulations under the present feature model; real-data validation here is limited to the Iridium $\beta = 0$ regime.

\emph{Algorithm design.}
Algorithm~\ref{alg:iwat} proceeds in two stages. Offline, $\beta(\calA)$ determines the active set $\calF$ and each feature receives weight $w_k = \DR_k^2 / \sum_{k'} \DR_{k'}^2$. The squaring emphasizes strong features: PA amplitude variance receives $w = 0.42$ while each IQ feature receives $w < 0.01$. Online, the test computes the weighted score

\begin{equation}\label{eq:iwat}
S_i = \sum_{k \in \calF} w_k \, (\bar{f}_k - \mu_{ik})^2,
\end{equation}
where $\bar{f}_k$ is the multi-message accumulated test feature (globally normalized) and $\mu_{ik}$ is the enrollment reference for satellite~$i$. The claimed identity is $i^* = \arg\min_i S_i$; the acceptance threshold $\tau$ is set by the target operating point on the enrollment-phase ROC curve. Both the DR weights $\{w_k\}$ and the threshold $\tau$ are computed exclusively from the enrollment data (File~A); no test-side information enters the offline stage.

\begin{algorithm}[t]
\caption{Identifiability-Weighted Authentication Test}
\label{alg:iwat}
\begin{algorithmic}[1]
\REQUIRE Constellation $\calA$, enrollment features $\{\hat{f}_{ik}^{(m)}\}$, test features $\{\hat{f}_{k}^{(m)}\}$
\ENSURE Accept / Reject
\STATE Compute $\beta(\calA) = 1 - |\E[x^2]|^2$
\IF{$\beta = 0$}
    \STATE $\calF \leftarrow$ \{PA, oscillator features\}
\ELSE
    \STATE $\calF \leftarrow$ \{PA, oscillator, IQ features\}
\ENDIF
\STATE Compute $\DR_k$ for $k \in \calF$ via balanced bootstrap
\STATE $w_k \leftarrow \DR_k^2 \,/\, \sum_{k'} \DR_{k'}^2$
\STATE Accumulate test fingerprint: $\bar{f}_k = N_{\mathrm{msg}}^{-1} \sum_m \hat{f}_k^{(m)}$
\STATE $S_i = \sum_{k \in \calF} w_k (\bar{f}_k - \mu_{ik})^2$ for each enrolled satellite $i$
\STATE $i^* = \arg\min_i S_i$
\STATE Set $\tau$ from enrollment ROC at target operating point
\IF{$S_{i^*} < \tau$}
    \STATE \textbf{Accept} as satellite $i^*$
\ELSE
    \STATE \textbf{Reject}
\ENDIF
\end{algorithmic}
\end{algorithm}

\emph{Results.}
DR$^2$-weighted IWAT achieves AUC~$= 0.934$, matching the hand-selected CRB-guided subset despite using six features including IQ (Table~\ref{tab:auth}, Fig.~\ref{fig:iwat}). Equal weighting yields AUC~$= 0.807$ ($-0.127$), confirming that the identifiability-informed discrimination structure provides effective emphasis for multi-feature fusion.

\begin{figure*}[t]
\centering
\includegraphics[width=0.9\textwidth]{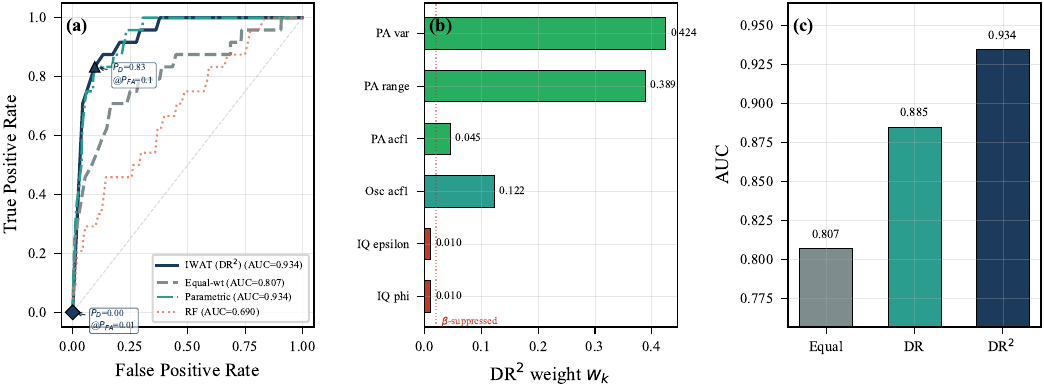}
\caption{IWAT authentication results. (a)~ROC comparison: DR$^2$-weighted IWAT matches CRB-guided hand selection (AUC~$= 0.934$), while equal-weight scoring on the same features yields AUC~$= 0.807$. (b)~DR$^2$ feature weights: PA features receive $>$80\% of total weight; IQ features are automatically suppressed to $w < 0.01$ ($\beta$-suppressed). (c)~Weight ablation: DR$^2$ weighting provides $+0.127$ AUC gain over equal weighting.}

\label{fig:iwat}
\end{figure*}

\section{Discussion}
\label{sec:discussion}

\emph{Validation scope.}
The $\beta = 0$ regime is supported by theorem, Monte Carlo, and real-data
validation; the $\beta = 1$ (QPSK) regime by theorem and Monte Carlo
(Fig.~\ref{fig:mc_crb}); OFDM/QAM predictions remain model-based only.

\begin{table}[h]
\centering
\small
\caption{Validation Status by Modulation Regime}
\label{tab:validation_scope}
\begin{tabular}{lccc}
\toprule
Regime & Theorem & Monte Carlo & Real Data \\
\midrule
$\beta=0$ (BPSK/preamble) & \checkmark & \checkmark & \checkmark \\
$\beta=1$ (QPSK) & \checkmark & \checkmark & --- \\
$\beta=1$ (16/64-QAM, OFDM) & \checkmark & --- & --- \\
\bottomrule
\end{tabular}
\end{table}

\emph{Iridium known-symbol structure.}
The 76~known symbols per burst (preamble + BPSK unique word) yield $\beta = 0$, reinforcing IQ unidentifiability. The $\beta = 1$ analysis applies to unknown data symbols and provides forward-looking predictions for OFDM systems.

\emph{Why IQ imbalance features fail.}
For constant-modulus DQPSK, IQ imbalance confines its effect to phase rotation indistinguishable from the channel ($\DR_\varepsilon < 1$), and the gr-iridium receiver~\cite{gr-iridium} partially compensates it through synchronization.

\emph{Scope and limitations of the CRB analysis.}
The CRB in Theorem~\ref{thm:fim} bounds estimation precision from raw baseband observations with known symbols. The experimental features in Section~\ref{sec:experiments} are extracted after symbol-rate decimation, CFO removal, and amplitude normalization by the gr-iridium pipeline~\cite{gr-iridium}. Two properties connect the bound to practice. First, the \emph{identifiability direction is preserved}: if a parameter is unidentifiable from raw observations ($\beta = 0$), no deterministic post-processing can create identifiability; the FIM null space is invariant under invertible transformations. This validates Prediction~2 and strengthens the identifiability-based explanation for OrbID. Second, for identifiable parameters ($\beta > 0$), the feature extraction pipeline may lose information relative to the raw CRB, so the bound is optimistic. This gap is consistent with the observed AUC~$= 0.93$ versus SatIQ's AUC~$= 0.96$, which operates on oversampled waveforms. Characterizing the exact information loss of symbol-rate processing remains an important open problem.

To quantify the pipeline dependence, we applied identical feature extraction to the SatIQ dataset~\cite{smailes2025satiq}, which provides oversampled (25~MS/s) IQ captures decimated to symbol rate via matched filtering. Despite originating from the same Iridium constellation, all features yield DR~$< 0.2$ (vs.\ DR~$> 4$ from PAST-AI), confirming that the gr-iridium decoder preserves PA-induced amplitude distortion that brute-force decimation destroys.

\emph{Comparison with existing methods.}

\begin{table*}[t]
\centering
\caption{Comparison of Satellite RF Fingerprinting Approaches$^\star$}
\label{tab:comparison}
\begin{tabular}{lcccc}
\toprule
 & SatIQ~\cite{smailes2025satiq} & PAST-AI~\cite{oligeri2023pastai} & ML baseline & This work \\
\midrule
Approach       & Neural         & CNN           & SVM/RF      & Parametric \\
Input          & 22k-dim IQ     & Symbol-rate   & 6 features  & 6 features (DR$^2$-wt) \\
Reported metric & AUC = 0.96    & Acc.\ = 85\%  & AUC $\leq$ 0.69 & AUC = 0.93 \\
Interpretable  & No             & No            & No          & Yes \\
Bounds         & None           & None          & None        & CRB \\
Failure pred.  & No             & No            & No          & Yes (OrbID) \\
\bottomrule
\multicolumn{5}{l}{\footnotesize $^\star$Metrics, inputs, and datasets differ across SatIQ and PAST-AI; ML baseline uses same data as this work.}
\end{tabular}
\end{table*}

Table~\ref{tab:comparison} compares our approach with SatIQ, PAST-AI, and the ML baselines. IWAT achieves AUC~$= 0.93$ using DR$^2$-weighted scoring on six features (Table~\ref{tab:auth}), combining multi-message accumulation ($\sqrt{N_{\mathrm{msg}}} \approx 22\times$ noise reduction) with identifiability-guided DR-weighted scoring. The SatIQ comparison is indirect as it uses oversampled captures from a different receiver.

\emph{Implications for next-generation satellites (model prediction).}
OFDM-based systems (Starlink, 5G-NTN) use 16/64-QAM with $\beta = 1$ and $\mu_6 > 1$. The CRB framework \emph{predicts} restored IQ identifiability and enhanced discrimination for these systems, with lower PA CRBs than PSK due to $\mu_6 > 1$ (Table~\ref{tab:beta}). This remains unvalidated on real non-Iridium data. The Monte Carlo validation (Fig.~\ref{fig:mc_crb}) confirms efficient estimation of all four parameters for QPSK under the current signal model. Extending the model to wideband OFDM channels and oversampled receivers is needed before these predictions can be considered operational.

\section{Conclusion}
\label{sec:conclusion}

We presented, to the best of our knowledge, the first estimation-theoretic performance analysis for satellite RF fingerprinting. The closed-form FIM (Theorem~\ref{thm:fim}) reveals that the IQ identifiability factor $\beta = 1 - |\E[x^2]|^2$ serves as a necessary condition for full IQ parameter recoverability: setting $\beta = 0$ for BPSK yields a rank-deficient FIM, showing that IQ fingerprinting is infeasible for binary-modulated satellites and providing a plausible explanation for OrbID's near-random performance on Orbcomm.

The discrimination metric correctly predicts that PA nonlinearity dominates for constant-modulus signals, validated with $r = 0.999$ cross-file stability on 24~Iridium satellites. The proposed IWAT algorithm converts these predictions into a practical authentication scoring rule: DR$^2$-weighting automatically suppresses uninformative features, achieving AUC~$= 0.934$ from six features (equal weighting: $0.807$; ML baselines: $\leq 0.69$). IQ-based authentication remains at chance level (AUC~$\approx 0.5$) regardless of accumulation, confirming the $\beta = 0$ identifiability limitation.

Future work includes extending the model to phase noise and filter response parameters, collecting oversampled data for full CRB validation across the Iridium constellation, and applying the framework to OFDM-based systems such as Starlink and 5G-NTN where the restored IQ identifiability ($\beta = 1$) and enhanced PA precision ($\mu_6 > 1$) predict improved discrimination.

\appendices

\section{Proof of Theorem~\ref{thm:fim}}
\label{app:thm1}

Substitute~\eqref{eq:rx_signal} into~\eqref{eq:fim_entry}. The derivatives of $K_1, K_2$ from~\eqref{eq:K1K2} yield sensitivity vectors $\delta_\varepsilon = \tfrac{e^{j\varphi}}{2} x - \tfrac{e^{-j\varphi}}{2} x^*$ and $\delta_\varphi = \tfrac{j(1+\varepsilon)e^{j\varphi}}{2} x + \tfrac{j(1+\varepsilon)e^{-j\varphi}}{2} x^*$ at $\alpha_3 = 0$, while $\delta_{\alpha_{3R}} = |x_{\mathrm{IQ}}|^2 x_{\mathrm{IQ}}$ and $\delta_{\alpha_{3I}} = j|x_{\mathrm{IQ}}|^2 x_{\mathrm{IQ}}$.

For i.i.d.\ symbols uniform on $\calA$, the FIM entries reduce to constellation moments $\mu_2 = \E[|x|^2] = 1$, $\mu_{20} = \E[x^2]$, $\mu_4 = \E[|x|^4]$, and $\mu_6 = \E[|x|^6]$. Evaluating the IQ diagonal entries:
\begin{align}
\E[|\delta_\varepsilon|^2] &= \tfrac{1}{2}(1 - \Re\{e^{-2j\varphi}\mu_{20}\}) = \beta_\varepsilon, \label{eq:Jee} \\
\E[|\delta_\varphi|^2] &= \tfrac{(1+\varepsilon)^2}{2}(1 + \Re\{e^{-2j\varphi}\mu_{20}\}) = (1{+}\varepsilon)^2\beta_\varphi. \label{eq:Jpp}
\end{align}
For circular constellations ($\mu_{20}=0$): $\beta_\varepsilon = \beta_\varphi = 1/2$. For BPSK ($\mu_{20}=1$, $\varphi$ small): $\beta_\varepsilon \approx 2\varphi^2 \approx 0$ while $\beta_\varphi \approx 1$, consistent with the large $[\bJ]_{\varphi\varphi}$ in Remark~\ref{rem:paradox}. The IQ cross-term $\E[\delta_\varepsilon^* \delta_\varphi]$ vanishes when $\mu_{20}=0$ and equals $\tfrac{1{+}\varepsilon}{2}\Im\{e^{-2j\varphi}\mu_{20}\}$ in general. For the PA sub-block, $\E[|\delta_{\alpha_{3R}}|^2] = \E[|x|^6] = \mu_6$ in the small-impairment limit, giving $\bJ_{\mathrm{PA}} = \mu_6 \bI_2$. For constant-modulus constellations, $\mu_6 = 1$. The PA-IQ cross-coupling at $\varepsilon = 0$ gives $\E[\delta_\varepsilon^* \cdot |x|^2 x] = \tfrac{e^{-j\varphi}}{2}\mu_4$ for circular constellations ($\E[|x|^2 x^2] = 0$ by symmetry). Collecting real parts yields~\eqref{eq:fim_block}--\eqref{eq:fim_cross}. \hfill$\square$

\bibliographystyle{IEEEtran}
\bibliography{refs}

@inproceedings{smailes2023watchthisspace,
  author    = {Smailes, Joshua and K\"{o}hler, Sebastian and Birnbach, Simon and Strohmeier, Martin and Martinovic, Ivan},
  title     = {{Watch This Space}: Securing Satellite Communication through Resilient Transmitter Fingerprinting},
  booktitle = {Proc. ACM SIGSAC Conf. Comput. Commun. Security (CCS)},
  year      = {2023},
  pages     = {608--621},
  doi       = {10.1145/3576915.3623135}
}

@article{smailes2025satiq,
  author    = {Smailes, Joshua and K\"{o}hler, Sebastian and Birnbach, Simon and Strohmeier, Martin and Martinovic, Ivan},
  title     = {{SatIQ}: Extensible and Stable Satellite Authentication using Hardware Fingerprinting},
  journal   = {ACM Trans. Privacy Security},
  year      = {2025},
  volume    = {29},
  number    = {1},
  articleno = {2},
  doi       = {10.1145/3768619}
}

@article{oligeri2023pastai,
  author    = {Oligeri, Gabriele and Sciancalepore, Savio and Raponi, Simone and Di Pietro, Roberto},
  title     = {{PAST-AI}: Physical-Layer Authentication of Satellite Transmitters via Deep Learning},
  journal   = {IEEE Trans. Inf. Forensics Security},
  year      = {2023},
  volume    = {18},
  pages     = {274--289},
  doi       = {10.1109/TIFS.2022.3219287}
}

@article{oligeri2023iridiumdata,
  author    = {Oligeri, Gabriele and Sciancalepore, Savio and Di Pietro, Roberto},
  title     = {Physical-Layer Data of {IRIDIUM} Satellites Broadcast Messages},
  journal   = {Data in Brief},
  year      = {2023},
  volume    = {46},
  pages     = {108905},
  doi       = {10.1016/j.dib.2023.108905}
}

@inproceedings{solenthaler2025orbid,
  author    = {Solenthaler, C\'{e}dric and Smailes, Joshua and Strohmeier, Martin},
  title     = {{OrbID}: Identifying Orbcomm Satellite {RF} Fingerprints},
  booktitle = {Proc. Workshop Security Space Satellite Syst. (SpaceSec)},
  year      = {2025},
  doi       = {10.14722/spacesec.2025.23031}
}

@article{zhang2025fingerprinting_survey,
  author    = {Zhang, Junqing and Ardizzon, Francesco and Piana, Mattia and Shen, Guanxiong and Tomasin, Stefano},
  title     = {Physical Layer-Based Device Fingerprinting for Wireless Security: From Theory to Practice},
  journal   = {IEEE Trans. Inf. Forensics Security},
  year      = {2025},
  volume    = {20},
  pages     = {5296--5325},
  doi       = {10.1109/TIFS.2025.3570118}
}

@inproceedings{smailes2024stickyfingers,
  author    = {Smailes, Joshua and Salkield, Edd and K\"{o}hler, Sebastian and Birnbach, Simon and Strohmeier, Martin and Martinovic, Ivan},
  title     = {{Sticky Fingers}: Resilience of Satellite Fingerprinting against Jamming Attacks},
  booktitle = {Proc. 2nd Workshop Security Space Satellite Syst. (SpaceSec)},
  year      = {2024},
  doi       = {10.14722/spacesec.2024.23054}
}

@book{kay1993fundamentals,
  author    = {Kay, Steven M.},
  title     = {Fundamentals of Statistical Signal Processing: Estimation Theory},
  volume    = {1},
  publisher = {Prentice Hall},
  year      = {1993}
}

@book{vantrees1968detection,
  author    = {{Van Trees}, Harry L.},
  title     = {Detection, Estimation, and Modulation Theory, {Part I}: Detection, Estimation, and Linear Modulation Theory},
  publisher = {John Wiley \& Sons},
  year      = {1968}
}

@book{cramer1946mathematical,
  author    = {Cram\'{e}r, Harald},
  title     = {Mathematical Methods of Statistics},
  publisher = {Princeton Univ. Press},
  year      = {1946}
}

@article{polak2011identifying,
  author    = {Polak, Adam C. and Dolatshahi, Sepideh and Goeckel, Dennis L.},
  title     = {Identifying Wireless Users via Transmitter Imperfections},
  journal   = {IEEE J. Sel. Areas Commun.},
  volume    = {29},
  number    = {7},
  pages     = {1469--1479},
  year      = {2011},
  doi       = {10.1109/JSAC.2011.110812}
}

@inproceedings{sankhe2019oracle,
  author    = {Sankhe, Kunal and Belgiovine, Mauro and Zhou, Fan and Riyaz, Shamnaz and Ioannidis, Stratis and Chowdhury, Kaushik},
  title     = {{ORACLE}: Optimized Radio clAssification through Convolutional neuraL nEtworks},
  booktitle = {Proc. IEEE INFOCOM},
  pages     = {370--378},
  year      = {2019},
  doi       = {10.1109/INFOCOM.2019.8737463}
}

@inproceedings{robyns2017physical,
  author    = {Robyns, Pieter and Marin, Eduard and Lamotte, Wim and Quax, Peter and Singel\'{e}e, Dave and Preneel, Bart},
  title     = {Physical-layer fingerprinting of {LoRa} devices using supervised and zero-shot learning},
  booktitle = {Proc. ACM Conf. Security Privacy Wireless Mobile Netw. (WiSec)},
  pages     = {58--63},
  year      = {2017},
  doi       = {10.1145/3098243.3098267}
}

@article{jian2020deep,
  author    = {Jian, Tong and Rendon, Bruno Costa and Ojuba, Emmanuel and Soltani, Nasim Yahya and Wang, Zifeng and Sankhe, Kunal and Gritsenko, Andrey and Dy, Jennifer G. and Chowdhury, Kaushik Roy and Ioannidis, Stratis},
  title     = {Deep Learning for {RF} Fingerprinting: A Massive Experimental Study},
  journal   = {IEEE Internet Things Mag.},
  volume    = {3},
  number    = {1},
  pages     = {50--57},
  year      = {2020},
  doi       = {10.1109/IOTM.0001.1900065}
}

@article{jagannath2022comprehensive,
  author    = {Jagannath, Anu and Jagannath, Jithin and Kumar, Prem Sagar Pattanshetty Vasanth},
  title     = {A comprehensive survey on radio frequency ({RF}) fingerprinting: Traditional approaches, deep learning, and open challenges},
  journal   = {Comput. Netw.},
  year      = {2022},
  volume    = {219},
  pages     = {109455},
  doi       = {10.1016/j.comnet.2022.109455}
}

@book{schenk2008rf,
  author    = {Schenk, Tim C. W.},
  title     = {{RF} Imperfections in High-Rate Wireless Systems: Impact and Digital Compensation},
  publisher = {Springer},
  year      = {2008},
  doi       = {10.1007/978-1-4020-6903-1}
}

@article{boulogeorgos2016effects,
  author    = {Boulogeorgos, Alexandros-Apostolos A. and Sofotasios, Paschalis C. and Selim, Bassant and Muhaidat, Sami and Karagiannidis, George K. and Valkama, Mikko},
  title     = {Effects of {RF} Impairments in Communications Over Cascaded Fading Channels},
  journal   = {IEEE Trans. Veh. Technol.},
  volume    = {65},
  number    = {11},
  pages     = {8878--8894},
  year      = {2016},
  doi       = {10.1109/TVT.2016.2516901}
}

@article{anttila2008circularity,
  author    = {Anttila, Lauri and Valkama, Mikko and Renfors, Markku},
  title     = {Circularity-Based {I/Q} Imbalance Compensation in Wideband Direct-Conversion Receivers},
  journal   = {IEEE Trans. Veh. Technol.},
  volume    = {57},
  number    = {4},
  pages     = {2099--2113},
  year      = {2008},
  doi       = {10.1109/TVT.2007.909269}
}

@inproceedings{rapp1991effects,
  author    = {Rapp, C.},
  title     = {Effects of {HPA}-Nonlinearity on a 4-{DPSK}/{OFDM}-Signal for a Digital Sound Broadcasting System},
  booktitle = {Proc. 2nd European Conf. Satellite Commun. (ECSC-2)},
  pages     = {179--184},
  year      = {1991}
}

@article{saleh1981frequency,
  author    = {Saleh, Adel A. M.},
  title     = {Frequency-Independent and Frequency-Dependent Nonlinear Models of {TWT} Amplifiers},
  journal   = {IEEE Trans. Commun.},
  volume    = {29},
  number    = {11},
  pages     = {1715--1720},
  year      = {1981},
  doi       = {10.1109/TCOM.1981.1094911}
}

@misc{gr-iridium,
  author       = {{Chaos Computer Club M\"{u}nchen}},
  title        = {gr-iridium: {Iridium} Burst Detector and Demodulator},
  year         = {2024},
  howpublished = {\url{https://github.com/muccc/gr-iridium}}
}

@article{luo2023ambiguity,
  author    = {Luo, Ruidan and Chen, Xiao and Wu, Zhongwang and Li, Yaping},
  title     = {Ambiguity function analysis of iridium burst transmissions for opportunistic navigation},
  journal   = {Electron. Lett.},
  year      = {2023},
  volume    = {59},
  number    = {21},
  pages     = {e12989},
  doi       = {10.1049/ell2.12989}
}

@inproceedings{maine1995iridium,
  author    = {Maine, K. and Devieux, C. and Swan, P.},
  title     = {Overview of {IRIDIUM} satellite network},
  booktitle = {Proc. WESCON},
  pages     = {483--490},
  year      = {1995},
  doi       = {10.1109/WESCON.1995.485428}
}

@inproceedings{danev2010attacks,
  author    = {Danev, Boris and Luecken, Heinrich and Capkun, Srdjan and {El Defrawy}, Karim},
  title     = {Attacks on Physical-layer Identification},
  booktitle = {Proc. ACM Conf. Wireless Netw. Security (WiSec)},
  year      = {2010},
  pages     = {89--98},
  doi       = {10.1145/1741866.1741882}
}

@techreport{3gpp_tr38811,
  author      = {{3GPP}},
  title       = {Study on {New Radio} ({NR}) to support non-terrestrial networks},
  institution = {3rd Generation Partnership Project},
  number      = {TR 38.811 v15.4.0},
  year        = {2020}
}

@article{loo1985statistical,
  author    = {Loo, C.},
  title     = {A statistical model for a land mobile satellite link},
  journal   = {IEEE Trans. Veh. Technol.},
  volume    = {VT-34},
  number    = {3},
  pages     = {122--127},
  year      = {1985},
  doi       = {10.1109/T-VT.1985.24048}
}

@article{bhattacharyya1943measure,
  author    = {Bhattacharyya, A.},
  title     = {On a measure of divergence between two statistical populations defined by their probability distributions},
  journal   = {Bull. Calcutta Math. Soc.},
  volume    = {35},
  pages     = {99--109},
  year      = {1943}
}

@book{casella2002statistical,
  author    = {Casella, George and Berger, Roger L.},
  title     = {Statistical Inference},
  edition   = {2nd},
  publisher = {Duxbury/Thomson Learning},
  year      = {2002}
}

@article{xu2016device,
  author    = {Q. Xu and R. Zheng and W. Saad and Z. Han},
  title     = {Device Fingerprinting in Wireless Networks: Challenges and Opportunities},
  journal   = {IEEE Commun. Surveys Tuts.},
  volume    = {18},
  number    = {1},
  pages     = {94--104},
  year      = {2016},
  doi       = {10.1109/COMST.2015.2476338}
}

@article{senigagliesi2021comparison,
  author    = {Senigagliesi, Linda and Baldi, Marco and Gambi, Ennio},
  title     = {Comparison of Statistical and Machine Learning Techniques for Physical Layer Authentication},
  journal   = {IEEE Trans. Inf. Forensics Security},
  volume    = {16},
  pages     = {1506--1521},
  year      = {2021},
  doi       = {10.1109/TIFS.2020.3033454}
}

@article{Zhang2021RFFI,
  author   = {J. Zhang and R. Woods and M. Sandell and M. Valkama and A. Marshall and J. Cavallaro},
  title    = {Radio Frequency Fingerprint Identification for Narrowband Systems, Modelling and Classification},
  journal  = {IEEE Trans. Inf. Forensics Security},
  volume   = {16},
  pages    = {3974--3987},
  year     = {2021},
  doi      = {10.1109/TIFS.2021.3088008},
}

@article{Shen2022LoRa,
  author   = {G. Shen and J. Zhang and A. Marshall and J. R. Cavallaro},
  title    = {Towards Scalable and Channel-Robust Radio Frequency Fingerprint Identification for {LoRa}},
  journal  = {IEEE Trans. Inf. Forensics Security},
  volume   = {17},
  pages    = {774--787},
  year     = {2022},
  doi      = {10.1109/TIFS.2022.3152404},
}

@article{Shen2023LengthVersatile,
  author   = {G. Shen and J. Zhang and A. Marshall and M. Valkama and J. R. Cavallaro},
  title    = {Toward Length-Versatile and Noise-Robust Radio Frequency Fingerprint Identification},
  journal  = {IEEE Trans. Inf. Forensics Security},
  volume   = {18},
  pages    = {2355--2367},
  year     = {2023},
  doi      = {10.1109/TIFS.2023.3266626},
}

@article{Luo2025ChannelRobust5G,
  author   = {H. Luo and G. Li and A. Brighente and M. Conti and Y. Xing and A. Hu and X. Wang},
  title    = {Channel-Robust {RF} Fingerprint Identification for Multi-Antenna {5G} User Equipments},
  journal  = {IEEE Trans. Inf. Forensics Security},
  volume   = {20},
  pages    = {10761--10776},
  year     = {2025},

}

@article{Rajendran2022RFImpairment,
  author   = {S. Rajendran and Z. Sun},
  title    = {{RF} Impairment Model-Based {IoT} Physical-Layer Identification for Enhanced Domain Generalization},
  journal  = {IEEE Trans. Inf. Forensics Security},
  volume   = {17},
  pages    = {1285--1299},
  year     = {2022},
  doi      = {10.1109/TIFS.2022.3158553},
}

@book{Kay1998Detection,
  author    = {S. M. Kay},
  title     = {Fundamentals of Statistical Signal Processing: Detection Theory},
  publisher = {Prentice Hall},
  address   = {Upper Saddle River, NJ},
  year      = {1998},
  volume    = {2},
}

@book{Stoica2005Spectral,
  author    = {P. Stoica and R. L. Moses},
  title     = {Spectral Analysis of Signals},
  publisher = {Pearson Prentice Hall},
  address   = {Upper Saddle River, NJ},
  year      = {2005},
}

@book{Lehmann1998PointEstimation,
  author    = {E. L. Lehmann and G. Casella},
  title     = {Theory of Point Estimation},
  edition   = {2nd},
  publisher = {Springer-Verlag},
  address   = {New York},
  year      = {1998},
  series    = {Springer Texts in Statistics},
}

@article{Yue2023LEOSecurity,
  author   = {P. Yue and J. An and J. Zhang and J. Ye and G. Pan and S. Wang and P. Xiao and L. Hanzo},
  title    = {Low {Earth} Orbit Satellite Security and Reliability: Issues, Solutions, and the Road Ahead},
  journal  = {IEEE Commun. Surveys Tuts.},
  volume   = {25},
  number   = {3},
  pages    = {1604--1652},
  year     = {2023},
  doi      = {10.1109/COMST.2023.3296160},
}

@inproceedings{Topal2022LEOPLA,
  author    = {O. A. Topal and G. {Karabulut Kurt}},
  title     = {Physical Layer Authentication for {LEO} Satellite Constellations},
  booktitle = {Proc. IEEE Wireless Commun. Netw. Conf. (WCNC)},
  address   = {Austin, TX, USA},
  month     = apr,
  year      = {2022},
  pages     = {1952--1957},
  doi       = {10.1109/WCNC51071.2022.9771727},
}

@article{Lohan2021GNSSSpoofing,
  author   = {E. S. Lohan and R. T. Iacob and O. Enescu},
  title    = {A Survey of Spoofer Detection Techniques via Radio Frequency Fingerprinting with Focus on the {GNSS} Pre-Correlation Sampled Data},
  journal  = {Sensors},
  volume   = {21},
  number   = {9},
  pages    = {3012},
  year     = {2021},
  doi      = {10.3390/s21093012},
}

@inproceedings{Hanna2019PAFingerprint,
  author    = {S. S. Hanna and D. Cabric},
  title     = {Deep Learning Based Transmitter Identification using Power Amplifier Nonlinearity},
  booktitle = {Proc. IEEE Int. Conf. Computing, Networking and Commun. (ICNC)},
  address   = {Honolulu, HI, USA},
  month     = feb,
  year      = {2019},
  pages     = {674--680},
  doi       = {10.1109/ICCNC.2019.8685569},
  note      = {Best Paper Award},
}

@article{Liu2023mmWavePLA,
  author   = {Y. Liu and P. Zhang and J. Liu and Y. Shen and X. Jiang},
  title    = {Exploiting Fine-Grained Channel/Hardware Features for {PHY}-Layer Authentication in {MmWave} {MIMO} Systems},
  journal  = {IEEE Trans. Inf. Forensics Security},
  volume   = {18},
  pages    = {4059--4074},
  year     = {2023},

}

@article{Xie2021MultiFeaturePLA,
  author   = {N. Xie and J. Chen and L. Huang},
  title    = {Physical-Layer Authentication Using Multiple Channel-Based Features},
  journal  = {IEEE Trans. Inf. Forensics Security},
  volume   = {16},
  pages    = {2356--2366},
  year     = {2021},
  doi      = {10.1109/TIFS.2021.3054534},
}

@article{Zhang2021PhaseNoisePLA,
  author   = {P. Zhang and J. Liu and Y. Shen and X. Jiang},
  title    = {Exploiting Channel Gain and Phase Noise for {PHY}-Layer Authentication in Massive {MIMO} Systems},
  journal  = {IEEE Trans. Inf. Forensics Security},
  volume   = {16},
  pages    = {4265--4279},
  year     = {2021},
  doi      = {10.1109/TIFS.2020.3030267},
}

\end{document}